\documentclass{acmtrans2m}

\usepackage{color}
\usepackage{amstext}
\usepackage{psfig}
\usepackage{epsfig}

\newdef{definition}{Definition}
\newdef{lemma}{Lemma}
\newdef{theorem}{Theorem}

\newcommand{\R}{{\bf R}}

\newcommand{\names}{{\sf Labels}}

\newcommand{\fopr}[1]{\mbox{${\sf FO}(+,\ab \times,\ab <,\ab 0,\ab 1,\ab #1)$}}

\usepackage{graphics}

\newcommand{\vmax}{v_{\rm max}}
\newcommand{\ab}{\allowbreak}
\newcommand{\const}[1]{{\sf #1}}

\newcommand{\B}{{\sf B}}
\newcommand{\uB}[1]{{\sf B^+_{#1}}}
\newcommand{\bB}[1]{{\sf B^-_{#1}}}
\newcommand{\Ba}{{\sf \tau B}}
\newcommand{\Bm}{{\sf \partial B}}
\newcommand{\Br}{{\sf \rho B}}
\newcommand{\IC}{{\sf IC}}
\newcommand{\Cp}[1]{{\sf C^+_{#1}}}
\newcommand{\Cm}[1]{{\sf C^-_{#1}}}
\newcommand{\bead}[7]{\B(#1,\ab #2,\ab #3,\ab #4,\ab #5,\ab #6,\ab #7)}
\newcommand{\beadp}[3]{\B(#1,\ab #2,\ab #3)}

\newcommand{\ubead}[7]{\uB{}(#1,\ab #2,\ab #3,\ab #4,\ab #5,\ab #6,\ab #7)}
\newcommand{\bbead}[7]{\bB{}(#1,\ab #2,\ab #3,\ab #4,\ab #5,\ab #6,\ab #7)}

\newcommand{\ucone}[4]{\Cp{}(#1,\ab #2,\ab #3,\ab #4)}
\newcommand{\bcone}[4]{\Cm{}(#1,\ab #2,\ab #3,\ab #4)}

\newcommand{\uconep}[2]{\Cp{}(#1,\ab #2)}
\newcommand{\bconep}[2]{\Cm{}(#1,\ab #2)}

\newcommand{\ICs}[2]{\IC(#1,\ab #2)}

\newcommand{\fop}{${\sf FO}(+,\ab \times,\ab <,\ab 0,\ab 1)$}

 \title{A case study of the difficulty of quantifier elimination in constraint databases:
the alibi query in moving object databases.}

\author{
 \emph{Walied Othman}\footnote{walied.othman@uhasselt.be}
\emph{Bart Kuijpers}\footnote{bart.kuijpers@uhasselt.be}
 \emph{Rafael Grimson}\footnote{rafael.grimson@uhasselt.be}\\Theoretical Computer Science\\Hasselt University \& Transnational University of Limburg, Belgium
}

\date{}

\begin{abstract}
In the constraint database model, spatial and  spatio-temporal data are stored by
boolean combinations of polynomial equalities and inequalities over the real numbers.
The relational calculus
augmented with polynomial constraints  is the
standard first-order query language for constraint databases.
Although the expressive power of this query language has been studied extensively,
the difficulty of the efficient evaluation
of queries, usually involving some form of quantifier
elimination, has received considerably less attention.
The inefficiency of existing quantifier-elimination software and the intrinsic difficulty of quantifier elimination have proven to be a bottle-neck for for real-world implementations of constraint database
systems.

In this paper, we focus on a particular query, called the \emph{alibi query}, that was proposed  in the context of moving object databases and that asks whether two moving objects   whose positions are known at certain moments in time, could have possibly met, given certain speed constraints.
This query can be seen as a constraint database query and its evaluation relies on the elimination of a block of three existential quantifiers. Implementations of general purpose elimination algorithms, such as provided by QEPCAD, Redlog and Mathematica, are in the specific case, for practical purposes, too slow in answering the alibi query and fail completely  in the parametric case.

The main contribution of this paper is an   analytical solution
to the parametric alibi query, which can be used to answer this query in the specific case in
constant time. We also give an analytic solution to the alibi query at a fixed moment in time, which asks whether
two moving objects that are known at discrete moments in time could have met at a given moment in time, given some speed constraints.

The solutions we propose are based on geometric argumentation and they
illustrate the fact that some practical problems require  creative solutions, where at least in theory, existing systems could provide a solution.
 \end{abstract}


\begin{document}
\maketitle

\section{Introduction and summary}\label{sec:intro}

The framework of \emph{constraint databases} was introduced in
1990 by Kanellakis, Kuper and Revesz~\cite{kkr95} as an extension
of the relational database model that allows the use of  infinite,
but first-order definable relations rather than just finite
relations. It provides an elegant and powerful model for
applications that deal with infinite sets of points in some real
affine space $\R^n$, such as spatial and spatio-temporal databases~\cite{pvv_pods94}.
In the setting of the constraint model,
infinite relations in spatial or spatio-temporal databases are finitely represented as boolean
combinations of polynomial equalities and inequalities, which are
interpreted  over the real   numbers.

Various aspects of the constraint database model are well-studied by now.
For an overview of research results we refer to~\cite{cdbook} and
the textbook~\cite{revesz}. The relational calculus augmented
with polynomial constraints, or equivalently, first-order logic
over the reals augmented with relation predicates to address the
database relations $R_1,\ldots , R_m$, \fopr{R_1,\ldots,R_m}\ for
short, is the standard first-order query language for constraint
databases. The expressive power of first-order logic over the
reals, as a constraint database query language, has been studied
extensively~\cite{cdbook}.
It is well-known that first-order
constraint queries can be effectively evaluated~\cite{tarski,cdbook}.
However, the difficulty of the efficient evaluation
of first-order queries, usually involving some form of quantifier
elimination, has been largely neglected~\cite{bart-joos}. The existing constraint
database systems or prototypes, such as Dedale and Disco~\cite[Chapters 17 and 18]{cdbook} are based on general purpose
quantifier-elimination algorithms and are, in most cases,
restricted to work with linear data, i.e., they use first-order
logic over the reals without multiplication~\cite[Part
IV]{cdbook}.
Of the general purpose elimination algorithms~\cite{basu,collins,russian,hrs,renegar}, some are now available in
software packages such as QEPCAD~\cite{qepcad}, Redlog~\cite{redlog} and Mathematica~\cite{mathematica}.
 But the  intrinsic inefficiency of quantifier elimination and the inefficiency of
 its current implementations  represent a
bottle-neck for real-world implementations of constraint database
systems~\cite{bart-joos}.

In this paper, we focus on a case study of quantifier elimination in constraint databases.
Our example is the \emph{alibi query} in moving object databases, which was introduced and studied in the area of geographic information systems (GIS)~\cite{pfoser,egenhofer,hornsby,miller}.
 This query can be expressed in the constraint database formalism and, at least in theory could be
answered, both in the specific and the parametric case, by existing implementations of
quantifier elimination over the reals. The  evaluation of the alibi query adds up to the elimination of a block of three existential quantifiers.  It turns out that packages such as QEPCAD~\cite{qepcad}, Redlog~\cite{redlog} and Mathematica~\cite{mathematica}, can only solve
the alibi query in specific cases, with a running time that is not acceptable for
moving object database users. In the parametric case, these quantifier-elimination implementations fail miserably. The main contribution of this paper is a theoretic and practical solution to the alibi query in the parametric case (and thus the specific cases).

 The research on spatial databases, which started in the 1980s from
work in geographic information systems, was extended in the second
half of the 1990s to deal with spatio-temporal data. In this field,
one particular line of research, started by Wolfson, concentrates on
\emph{moving object databases} (MODs)~\cite{guting,wolfson}, a field
in which several data models and query languages have been proposed
to deal with moving objects whose position is recorded at discrete
moments in time. Some of these  models are geared towards handling
uncertainty that may come from various sources (measurements of
locations, interpolation, ...) and several query formalisms have
been proposed~\cite{su,floris,ons-icdt}. For an overview of models and
techniques for MODs, we refer to the book by G\"uting and
Schneider~\cite{guting}.

In this paper, we focus on the trajectories that are produced by
moving objects and which are stored in a database as a collection of
tuples $(\const {t}_i,\const{x}_i,\const{y}_i)$, $i=0,...,N$, i.e.,
as a finite sample of time-stamped locations in the plane. These
samples may have been obtained by GPS-measurements or from other
location aware devices.

One particular model for the management of the uncertainty of the
moving object's position in between sample points is provided by the
\emph{bead} model. In this model, it is assumed that besides the
time-stamped locations of the object also some background knowledge,
in particular a (e.g., physically or law imposed) speed limitation
$\const{v}_i$ at location $(\const{x}_i,\const{y}_i)$ is known. The
bead between two  consecutive sample points is defined as  the
collection of time-space points where the moving objects can have
passed, given the speed limitation (see Figure~\ref{necklace} for
an illustration). The  chain of beads connecting consecutive
trajectory sample points is called a \emph {lifeline
necklace}~\cite{egenhofer}. Whereas beads were already conceptually
known in the time geography of H\"agerstrand in the
1970s~\cite{hagerstrand}, they were introduced in the area of GIS by
Pfoser~\cite{pfoser} and later studied by
Egenhofer and Hornsby~\cite{egenhofer,hornsby}, and Miller~\cite{miller}, and in a query language
context by the present authors~\cite{ons-icdt}.

In this setting, a query of particular interest that has been
studied, mainly by Egenhofer and Hornsby~\cite{egenhofer,hornsby},
is the \emph{alibi query}. This boolean query asks whether two
moving objects, that are given by samples of time-space points and
speed limitations,  can have physically met. This question adds up
to deciding whether the necklaces of beads of these moving objects
intersect or not. This problem can be considered  solved in
practice, when we can efficiently decide whether two beads
intersect.

Although approximate solutions to this problem have been
proposed~\cite{egenhofer}, also an exact solution is possible. We
show that the alibi query can be formulated  in the  constraint
database model by means of a first-order
query.  Experiments with software packages such as
QEPCAD~\cite{qepcad}, Redlog~\cite{redlog}  and Mathematica~\cite{mathematica} on a variety
of beads show that deciding if two concrete beads intersect can be
computed on average in 2 minutes (running Windows XP Pro, SP2, with
a Intel Pentium M, 1.73GHz, 1GB RAM). This means that evaluating the
alibi query on the lifeline necklaces of two moving objects that
each consist of  100 beads would take around $100\times 100\times 2$
minutes, which is almost two weeks, or, if we work with ordered time intervals and first test on overlapping time intervals,  $(100+ 100)\times 2$ minutes, which is almost 7 hours. Clearly, such an amount of time
is unacceptable.

Another solution within the range of constraint databases is to find
a formula, in which the apexes and limit speeds of two beads appear
as parameters, that parametrically expresses that two beads
intersect. We call this problem the \emph{parametric alibi query}. A
quantifier-free formula for this parametric version could, in
theory,  also be obtained by eliminating one block of three
existential quantifiers from a formula with 17 variables using existing quantifier-elimination
packages.
We have attempted this approach  using Mathematica, Redlog and
QEPCAD, but after several days of running, with the configuration described above, we have interrupted the computation, without
successful outcome. Clearly, the
eliminating a block of three existential quantifiers from a formula in 17 variables is beyond the existing quantifier-elimination implementations.
In fact, it is well known that these implementations fail on complicated,
higher-dimensional   problems. The benefit of having a quantifier-free
first-order formula that expresses that two beads intersect is that
the alibi query on two beads can be answered in constant time. The
problem of deciding whether two lifeline necklaces intersect can then be
done in time proportional to the product (or the sum, if we first test on overlapping time intervals) of the lengths of the two
necklaces of beads.

The main contribution of this paper is the description of an
analytic solution to the alibi query. We give a quantifier-free
formula,  that contains square roots, however, and that expresses the
(non-)emptiness of the intersection of two parametrically given beads. Although, in a
strict sense, this formula cannot be seen as quantifier-free
first-order formula (due to the roots), it still gives the above
mentioned complexity benefits. Also, this formula with square roots can easily be turned into a quantifier-free formula of similar length. At the basis of our solution is a geometric theorem that
describes three exclusive cases in which beads can intersect. These three cases can then
be transformed into an analytic solution that can be used to answer the
alibi query on the lifeline necklaces of two moving objects in less than a minute.
This provides a practical solution to the alibi query.

To back up our claim that the execution time of our method requires milliseconds or less we implemented this in {\sc Mathematica} and compared it to using traditional quantifier elimination to decide this query. We have included this implementation in the Appendix and used it to perform numerous experiments which only confirm our claims.

We give another example of a problem where common sense prevails over the existing implementations of general quantifier elimination methods.
This problem is the \emph{alibi query at a fixed moment in time}, which asks whether
two moving objects that are known at discrete moments in time could have met at a given moment in time. This problem can be translated in deciding whether four disks
in the two-dimsnional plane
have a non-empty intersection. Again, this problem can be formulated in the context of the contraint model and adds up to the elimination of a block of two existential quantifiers.
Also for this problem we provide an exact solution in terms of a quantifier-free formula.

This paper is organized as follows.
In Section~\ref{sec:dbase}, we describe a model for trajectory (or  moving object)
databases with uncertainty using beads.
   In Section~\ref{sec:alibiquery}, we discuss the alibi query. The geometry of beads is
   discussed in Section~\ref{sec:prelim-beads}. An analytic
solution to this query is given in Section~\ref{sec:2dimalibi} and experimental results in Section~\ref{sec:experiments} of our implementation that can be found in
the Appendix.
The alibi query at a fixed moment in time is solved in Section~\ref{sec:4circles}.

\section{A model for moving object data with uncertainty}\label{sec:dbase}

In this paper, we consider moving objects in the two-dimensional
$(x,y)$-space $\R^2$ and describe their movement in the
$(t,x,y)$-space $\R\times\R^2$, where $t$ is time (we denote the set of the real numbers by $\R$).

In this section, we define trajectories, trajectory samples, beads and trajectory (sample) databases.
Although it is more traditional to speak about moving object databases, we
use the term trajectory databases to emphasize that we manage the trajectories  produced by moving objects.

\subsection{Trajectories and trajectory samples}

Moving objects, which we assume to be points, produce a special kind of curves, which
are parameterized by time and which we call \emph{trajectories}.

 \begin{definition}\rm
 A \emph{trajectory} $T$ is the graph of  a mapping $I\subseteq
\R\rightarrow \R^2: t\mapsto \alpha(t)=(\alpha_x(t), \alpha_y(t)), $
i.e., $$T=\{(t,\alpha_x(t),\alpha_y(t))\in\R\times\R^2\mid t\in I
\},$$ where $I$ is   the
\emph{time domain} of $T$.
\qed
\end{definition}

In practice, trajectories   are only known at discrete moments in time.  This partial knowledge of  trajectories is formalized in the following definition.
If we want to stress that some $t,x,y$-values (or other values) are constants, we will use sans serif characters.

\begin{definition}\rm
A \emph{trajectory sample}  is a finite set of time-space points $\{
(\const{t}_0, \ab \const{x}_0,\ab \const{y}_0),\ab (\const{t}_1, \ab \const{x}_1,\ab \const{y}_1),...,\ab (\const{t}_N, \ab
\const{x}_N,\ab \const{y}_N)\}$, on which  the order on time, $\const{t}_0<\const{t}_1<\cdots <\const{t}_N$, induces a natural order. \qed
 \end{definition}

For practical purposes, we may assume that the $(\const{t}_i,\const{x}_i,\const{y}_i)$-tuples of a trajectory sample contain rational values.

A trajectory $T$, which contains a  trajectory sample $\{(\const{t}_0, \ab \const{x}_0,\ab \const{y}_0),\ab (\const{t}_1, \ab \const{x}_1,\ab \const{y}_1),...,\ab (\const{t}_N, \ab
\const{x}_N,\ab \const{y}_N)\}$, i.e., $(\const{t}_i,\alpha_x(\const{t}_i),\alpha_y(\const{t}_i))=(\const{t}_i,\const{x}_i,\const{y}_i)$ for $i=0,...,N$,
is called a \emph{geospatial lifeline} for this trajectory sample~\cite{egenhofer}.
A common example of a lifeline, is the reconstruction of a trajectory from a trajectory samples by linear interpolation~\cite{guting}.

\subsection{Modeling uncertainty with beads}

Often, in practical applications, more is known about trajectories than merely some sample points
$(\const{t}_i,\const{x}_i,\const{y}_i)$. For instance, background knowledge like a physically or law imposed speed limitation $\const{v}_i$ at location $(\const{x}_i,\const{y}_i)$ might be available. Such a speed limit might even depend on $\const{t}_i$.
The speed limits that hold between two consecutive sample points can be used to model the uncertainty of a moving object's location between sample points.

More specifically, we know that at a time $t$,  $\const{t}_i\leq
t\leq \const{t}_{i+1}$, the object's distance to $(\const{x}_i,\const{y}_i)$  is at most
$\const{v}_i(t-\const{t}_i)$ and
its distance to
$(\const{x}_{i+1},\const{y}_{i+1})$ is at most
$\const{v}_i(\const{t}_{i+1}-t)$.
The spatial location of the  object is therefore somewhere in the intersection of the disc
with center  $(\const{x}_i,\const{y}_i)$ and radius $\const{v}_i(t-\const{t}_i)$ and the disc with center  $(\const{x}_{i+1},\const{y}_{i+1})$ and radius $\const{v}_i(\const{t}_{i+1}-t)$. The geometric location of these points is referred to as a \emph{bead}~\cite{pfoser,egenhofer} and defined, for general points $p=(t_p,x_p,y_p)$ and
$q=(t_q,x_q,y_q)$ and speed limit $\vmax$ as follows.

\begin{figure}[h]
\centering
\input{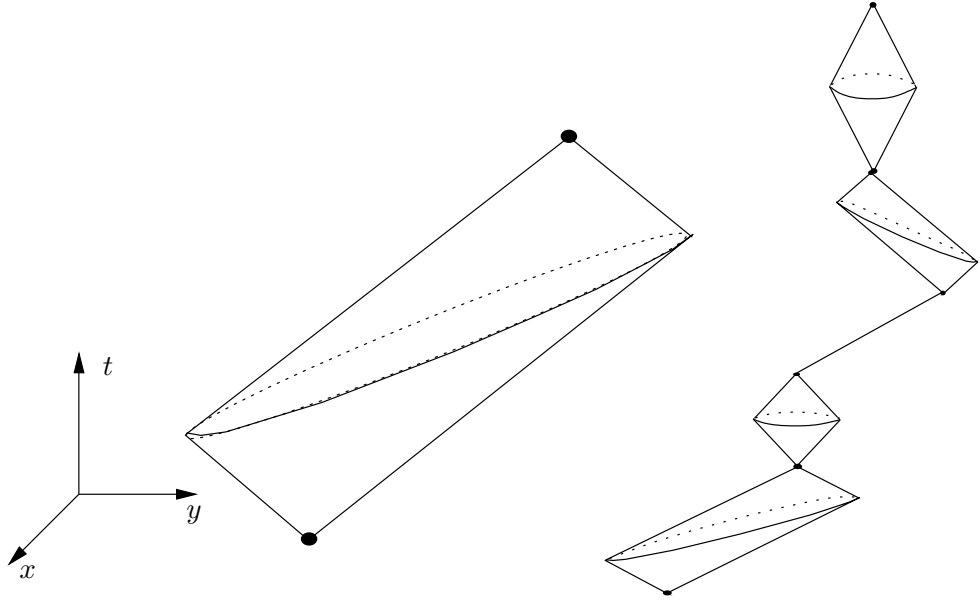}
\caption{A bead and a lifeline necklace.}\label{necklace}
\end{figure}

\begin{definition}\rm
The \emph{bead} with origin $p=(t_p,x_p,y_p)$, destination
$q=(t_q,x_q,y_q)$, with $t_p\leq t_q$, and maximal speed $\vmax\geq 0$ is the set of all points
$(t,x,y)\in\R\times\R^2$ that satisfy the following constraint formula\footnote{Later on, this type of formula's will be refered to as \fop-formulas.}
$$\displaylines{\quad \Psi_{\B}(t,x,y,t_p,x_p,y_p,t_q,x_q,y_q,\vmax)
:= (x-x_p)^2+(y-y_p)^2\leq (t-t_p)^2\vmax^2 \hfill{}\cr \hfill{}
\land \ (x-x_q)^2+(y-y_q)^2\leq (t_q-t)^2\vmax^2 \ \land \ t_p\leq
t\leq t_q.}
$$
We denote this set  by $\beadp{p}{q}{\vmax}$ or
$\bead{t_p}{x_p}{y_p}{t_q}{x_q}{y_q}{\vmax}$. \qed
 \end{definition}

 In the formula $\Psi_{\B}(t,x,y,t_p,x_p,y_p,t_q,x_q,y_q,\vmax)$, we consider $t_p,\ab x_p,\ab y_p,\ab t_q,\ab x_q,\ab y_q,\ab \vmax$ to be parameters, whereas $t,x,y$ are considered variables defining the subset of $\R\times\R^2$.

Figure~\ref{dissection} illustrates the notion of bead in
time-space. Whereas a continuous curve connecting the sample points
of a trajectory sample was called a geospatial lifeline, a chain of
beads connecting succeeding trajectory sample points is called a
\emph{lifeline necklace}~\cite{egenhofer}.

 \subsection{Trajectory databases}

We assume the
existence of an infinite set $\names =\{\const{a},\const{b}, ..., \const{a}_1,\const{b}_1,...,\const{a}_2,\const{b}_2,...\}$ of \emph{trajectory
labels}, that serve to identify individual trajectory samples.
We now define the notion of trajectory database.

\begin{definition}\rm
A \emph{trajectory (sample) database}
is a finite set of tuples $(\const{a}_i,\ab \const{t}_{i,j},\ab \const{x}_{i,j},\ab \const{y}_{i,j},\const{v}_{i,j})$, with $i=1,...,r$ and $j=0,...,N_i$,
such that $\const{a}_i\in \names$ cannot appear twice in combination with the same $t$-value,
such that $\{ (\const{t}_{i,0},\ab \const{x}_{i,0},\ab \const{y}_{i,0}), (\const{t}_{i,1},\ab \const{x}_{i,1},\ab \const{y}_{i,1}), \ab ..., (\const{t}_{i,N_i},\ab
\const{x}_{i,N_i},\ab \const{y}_{i,N_i})\} $ is a trajectory sample for each $i=1,...,r$ and such that the $\const{v}_{i,j}\geq 0$ for each $i=1,...,r$ and $j=0,...,N_i$.\qed
\end{definition}

\section{Trajectory queries and the alibi query}\label{sec:alibiquery}

In this section, we define the notion of trajectory database query, we show how constraint database languages can be used to query trajectories and we define the alibi query and the parametric alibi query.

 \subsection{Trajectory queries}

A \emph{trajectory database query} has been defined as a partial
computable function from trajectory databases  to trajectory
databases~\cite{ons-icdt}. Often, we are also interested in queries
that express a property, i.e., in boolean queries. More formally,
 we can say that a \emph{boolean trajectory database query} is a partial computable function
from trajectory databases to $\{{\sf True},{\sf False}\}$.

When we say that a function is computable, this is with respect to some
fixed encoding of the trajectory databases (e.g.,  rational numbers are represented as pairs of natural numbers in bit representation).

\subsection{A constraint-based query language}

Several languages have been proposed to express queries on moving object data and trajectory databases (see~\cite{guting} and references therein).
One particular language for querying trajectory data, that was recently studied in detail by the present authors,
is provided by the formalism of constraint databases. This query
language is a first-order logic
which extends first-order logic over the real numbers with a predicate $S$ to address the input trajectory database.
We denote this logic by \fopr{S} and define it  as follows.
\begin{definition}\rm
The language \fopr{S} is a two-sorted logic with \emph{label variables} $a, b, c,...$ (possibly with subscripts)
that refer to trajectory labels  and \emph{real variables} $x,y,z,..., v, ...$ (possibly with subscripts) that refer to real numbers.
The atomic formulas of \fopr{S} are
\begin{itemize}
\item $P(x_1,...,x_n)>0$, where $P$ is a polynomial with integer coefficients in the real variables $x_1,...,x_n$;
\item $a=b$; and
\item $S(a,t,x,y,v)$ ($S$ is a 5-ary predicate).
\end{itemize}
The formulas of \fopr{S} are built from the atomic formulas using the logical connectives
$\land, \lor,\lnot, ...$ and quantification over the two types of variables: $\exists x$, $\forall x$ and $\exists a$, $\forall a$.
\qed\end{definition}

The label variables are assumed to range over the labels occurring in the input trajectory database and the real variables are assumed to range over $\R$.
The formula $S(a,t,x,y,v)$ expresses that a tuple $(a,t,x,y,v)$ belongs to the input trajectory database. The interpretation of the other formulas is standard.

For example, the \fopr{S}-sentence
$$\exists {a}\exists {b}(\lnot ({a}={b})\land \forall {t}\forall {x}\forall {y} \forall {v} S({a},{t},{x},{y},{v})\leftrightarrow S({b},{t},{x},{y},{v})) $$
 expresses the boolean trajectory query that says that there are two identical trajectories in the input database with different labels.

When we instantiate the free variables in a \fopr{S}-formula $\varphi(a,\ab b,...,\ab t,\ab x,\ab y, ...)$
by concrete values $\const{a},\const{b},...,\const{t},\const{x},\const{y}, ...$ we write $\varphi[\const{a},\const{b},...,\const{t},\const{x},\const{y}, ...]$ for the formula we obtain.

\subsection{The alibi query}

The \emph{alibi query} is the boolean query which asks whether two moving objects, say with labels $\const{a}$ and $\const{a}'$, that are available as samples in a trajectory database, can have physically met.
Since the possible positions of these moving objects are, in between sample points, given by beads,
  the alibi query asks to decide if the two
lifeline necklaces of $\const{a}$ and $\const{a}'$ intersect or not.

More concretely, if the trajectory $\const{a}$ is given in the trajectory database by the tuples
$(\const{a}, \const{t}_0,\const{x}_0,\const{y}_0,\const{v}_0), ...., (\const{a}, \const{t}_N,\const{x}_N,\const{y}_N,\const{v}_N)$ and the trajectory $\const{a}'$ by the tuples
$(\const{a}', \const{t}'_0,\const{x}'_0,\const{y}'_0,\const{v}'_0), ...., (\const{a}',
\const{t}'_M,\const{x}'_M,\const{y}'_M,\const{v}'_M)$, then
$\const{a}$ has an alibi for not meeting $\const{a}'$ if for all $i$, $0\leq i\leq N-1$ and all $j$,
$0\leq j\leq M-1$, $$B(\const{t}_i,\const{x}_i,\const{y}_i,\const{t}_{i+1},\const{x}_{i+1}, \const{y}_{i+1},\const{v}_i)\cap
B(\const{t}'_j,\const{x}'_j,\const{y}'_j,\const{t}'_{j+1},\const{x}'_{j+1}, \const{y}'_{j+1},\const{v}'_j)=\emptyset.\eqno{(\dagger)}$$

We remark that the alibi query can be expressed by a formula in the logic \fopr{S},
which we know give. To start, we denote the subformula
$$\displaylines{\quad
S(a,\const{t}_1,\const{x}_1,\const{y}_1,\const{v}_1)\land S(a,t_2,x_2,y_2,v_2)
\land \hfill{}\cr\hfill{}\forall t_3\forall x_3\forall y_3\forall v_3
(
S(a,t_3,x_3,y_3,v_3)\rightarrow \lnot (t_1<t_3\land t_3<t_2)
),\quad
}$$
that expresses that
$(t_1,x_1,y_1)$ and $(t_2,x_2,y_2)$
are consecutive sample points on the trajectory $a$ by
$\sigma(a,t_1,x_1,y_1,v_1,t_2,x_2,y_2,v_2)$.

The alibi query on $\const{a}$ and $\const{a}'$ is then expressed as $\varphi_{\rm alibi}[\const{a},\const{a}']=$
$$
\displaylines{\quad
\lnot \exists t_1\exists x_1\exists y_1\exists v_1
\exists t_2\exists x_2\exists y_2\exists v_2
\exists t'_1\exists x'_1\exists y'_1\exists v'_1
\exists t'_2\exists x'_2\exists y'_2\exists v'_2\hfill{}\cr\qquad (
\sigma(\const{a},t_1,x_1,y_1,v_1,t_2,x_2,y_2,v_2)
\land \sigma(\const{a}',t'_1,x'_1,y'_1,v'_1,t'_2,x'_2,y'_2,v'_2)\land  \hfill{}\cr\quad \quad \     \exists t \exists x\exists y
(t_1\leq t\leq t_2\land t'_1\leq t\leq t'_2\land  \hfill{}\cr\qquad \quad
(x-x_1)^2+(y-y_1)^2\leq (t-t_1)^2v_1^2 \land
(x-x_2)^2+(y-y_2)^2\leq (t_2-t)^2v_1^2 \land   \hfill{}\cr\qquad\quad(x-x'_1)^2+(y-y'_1)^2\leq (t-t'_1)^2v'^2_1 \land
(x-x'_2)^2+(y-y'_2)^2\leq (t'_2-t)^2v'^2_1
)).\quad}
$$

It is well-known that \fopr{S}-expressible queries can be evaluated
effectively   on arbitrary trajectory database
inputs~\cite{cdbook,ons-icdt}. Briefly explained, this evaluation
can be performed by (1) replacing the occurrences of
$S(\const{a},t,x,y,v)$ by a disjunction describing all the sample
points belonging to the trajectory sample $\const{a}$; the same for
$\const{a}'$; and (2) eliminating all the quantifiers in the
obtained formula. In concreto, using the notation from above, each
occurrence of $S(\const{a},t,x,y,v)$ would be replaced in
$\varphi_{\rm alibi}[\const{a},\const{a}']$ by
$\bigvee_{i=0}^{N-1} (t=\const{t}_i\land x=\const{x}_i\land y=\const{y}_i\land v=\const{v}_i),$ and similar for $\const{a}'$.
This results in a (rather complicated) first-order formula over the
reals $\tilde\varphi_{\rm alibi}[\const{a},\const{a}']$ in which the
predicate $S$ does not occur any more. Since first-order logic over
the reals admits the elimination of quantifiers (i.e., every formula
can be equivalently expressed by a quantifier-free formula), we can
decide the truth value of $\tilde\varphi_{\rm
alibi}[\const{a},\const{a}']$ by eliminating all quantifiers from
this expression. In this case, we have to eliminate one block of
existential quantifiers.

We can however simplify the quantifier-elimination problem.
It is easy to see, looking at $(\dagger)$ above,
that $\lnot \tilde\varphi_{\rm alibi}[\const{a},\const{a}']$ is equivalent to
$$
\bigvee_{i=0}^{N-1} \bigvee_{j=0}^{M-1}
\psi_{alibi}[\const{t}_i,\const{x}_i,\const{y}_i,
\const{t}_{i+1},\const{x}_{i+1},\const{y}_{i+1}, \const{v}_i,
\const{t}'_j,\const{x}'_j,\const{y}'_j,\const{t}'_{j+1},\const{x}'_{j+1},\const{y}'_{j+1}, \const{v}'_j], $$
where the restricted alibi-query formula
$\psi_{alibi}(t_i,\ab x_i,\ab y_i,\ab t_{i+1},\ab x_{i+1},\ab y_{i+1}, v_i,
t'_j,\ab x'_j,\ab y'_j,\ab t'_{j+1},\ab x'_{j+1},\ab y'_{j+1},\ab  v'_j)$
abbreviates the formula
  $$\displaylines{\quad \exists t\exists x\exists y
(t_i\leq t\leq t_{i+1}\land t'_j\leq t\leq t'_{j+1}\land
(x-x_i)^2+(y-y_i)^2\leq (t-t_i)^2v_i^2 \land \hfill{}\cr \hfill{}
(x-x_{i+1})^2+(y-y_{i+1})^2\leq (t_{i+1}-t)^2v_i^2 \land
\hfill{}
\cr \hfill{} (x-x'_j)^2+(y-y'_j)^2\leq (t-t'_j)^2v'^2_j
\land(x-x'_{j+1})^2+(y-y'_{j+1})^2\leq (t'_{j+1}-t)^2v'^2_j )}$$
that expresses that two beads intersect.

So, the instantiated formula
$$\psi_{alibi}[\const{t}_i,\ab \const{x}_i,\ab \const{y}_i,\ab \const{t}_{i+1},\ab \const{x}_{i+1},\ab
\const{y}_{i+1}, \ab \const{v}_i, \ab \const{t}'_j,\ab
\const{x}'_j,\ab \const{y}'_j,\ab
\const{t}'_{j+1},\ab \const{x}'_{j+1},\ab \const{y}'_{j+1}, \ab
\const{v}'_j]$$  expresses   $(\dagger)$.
 To eliminate the existential block of quantifiers ($\exists t\exists x\exists y$)
 from this expression,  existing software-packages  for quantifier elimination,
 such as QEPCAD~\cite{qepcad}, Redlog~\cite{redlog} and Mathematica~\cite{mathematica} can be used.
We experimented QEPCAD, Redlog and Mathematica to decide if several beads intersected. The latter two programs have a similar performance and they outperform QEPCAD. To give an idea of their performance, we give some results with Mathematica:
the computation of $\psi_{alibi}[0,\ab 0,\ab 0,\ab 1,\ab 2,\ab 2,\ab
\sqrt{8},\ab 0,\ab 3,\ab 3,\ab 1,\ab 2,\ab 2,\ab 2]$ took $6$
seconds; that of  $\psi_{alibi}[0,\ab 0,\ab 0,\ab 1,\ab 2,\ab 2,\ab
\sqrt{8},\ab 0,\ab 3,\ab 4,\ab 1,\ab 2,\ab 2,\ab 2]$ took $209$
seconds and the computation of
$\psi_{alibi}[0,\ab 0,\ab 0,\ab 1,\ab -1,\ab -1,\ab
1,\ab 0,\ab 1,\ab 1,\ab 2,\ab -1,\ab 1,\ab 2]$ took $613$ seconds.
Roughly speaking, our experiments show that, using Mathematica , this quantifier elimination can be computed on average in about
2 minutes (running Windows XP Pro, SP2, with a Intel Pentium M,
1.73GHz, 1GB RAM). This means that evaluating the
alibi query on the lifeline necklaces of two moving objects that
each consist of  100 beads would take around $100\times 100\times 2$
minutes, which is almost two weeks, when applied naively and at most $(100+100)\times 2$ minutes or a quarter day, when first the intersection of time-intervals is tested.
Clearly, in both cases, such an amount of time
is unacceptable.

 There is a better solution, however, which we discuss next, that
 can decide if two beads intersect or not in a couple of milliseconds.

 \subsection{The parametric alibi query}
  The uninstantiated  formula  $$\psi_{alibi}(t_i,\ab x_i,\ab y_i,\ab t_{i+1},\ab x_{i+1},\ab y_{i+1}, v_i,
t'_j,\ab x'_j,\ab y'_j,\ab t'_{j+1},\ab x'_{j+1},\ab
y'_{j+1},\ab  v'_j)$$ can be viewed as a parametric version of the
restricted alibi query, where the free variables are considered
parameters. This formula contains three existential quantifiers and
the  existing soft\-ware-packages  for quantifier elimination could
be used to obtain a quantifier-free formula
$\tilde{\psi}_{alibi}(t_i,\ab x_i,\ab y_i,\ab t_{i+1},\ab
x_{i+1},\ab y_{i+1}, v_i, t'_j,\ab x'_j,\ab y'_j,\ab
t'_{j+1},\ab x'_{j+1},\ab y'_{j+1},\ab  v'_j)$
 that is equivalent to $\psi_{alibi}$. The formula $\tilde{\psi}_{alibi}$ could then be used to
straightforwardly  answer the alibi query in time linear in its size, which is  independent
of the size of the input and therefore constant.
 We have tried to eliminate the existential block of quantifiers $\exists t\exists x\exists y$
 from $\psi_{alibi}$ using Mathematica, Redlog  and QEPCAD. After some minutes of running, Redlog invokes QEPCAD. After several days of running QEPCAD on the
 configuration described above, we have interrupted the computation without result. Also Mathematica
 ran into problems without giving an answer. It is clear that eliminating a block of three existential quantifiers from a formula in 17 variables is beyond the existing quantifier-elimination implementations. Also, the instantiation of several parameters to adequately chosen constant values does not help to produce a solution.
For instance, without loss of generality we can
locate  $(t_i,\ab x_i,\ab y_i)$ in the origin $(0,0,0)$ and locate the other apex of the first bead above the $y$-axis, i.e., we can take $x_{i+1}=0$. Furthermore, we can take $v_i =1$ and $t_{i+1}=1$.
But Mathematica, Redlog  and QEPCAD cannot also not cope with this simplified situation.

 The main contribution of this paper is a the description of a quantifier-free formula equivalent to
 $\psi_{alibi}(t_i,\ab x_i,\ab y_i,\ab t_{i+1},\ab x_{i+1},\ab y_{i+1}, v_i,
t'_j,\ab x'_j,\ab y'_j,\ab t'_{j+1},\ab x'_{j+1},\ab y'_{j+1},\ab  v'_j)$.
The solution we give is not a quantifier-free first-order formula in a strict sense, since it contains
root expressions, but it can be easily turned into a quantifier-free first-order formula of
similar length.  It answers the alibi query on the lifeline necklaces of two moving objects that
each consist of  100 beads in less than a minute.
This description of this quantifier-free formula  is the subject of the next section.

\section{Preliminaries on the geometry of beads}\label{sec:prelim-beads}

Before, we can give an analytic solution to the alibi query and prove its correctness,
 we need to introduce some terminology concerning beads.

\subsection{Geometric components of  beads}

Various geometric properties of beads have already been
described~\cite{egenhofer,ons-icdt,miller}. Here, we need some more
definitions and notations to describe various components of a bead.
These components are illustrated in Figure~\ref{dissection}.
 In this section, let  $p=(t_p,x_p,y_p)$ and
$q=(t_q,x_q,y_q)$  be two time-space points, with $t_p\leq t_q$ and
let $\vmax$ be a positive real number.

The bead $\beadp{p}{q}{\vmax}$ is the intersection of two filled cones, given by the
equations
$(x-x_p)^2+(y-y_p)^2 \leq (t-t_p)^2\vmax^2 \ \land \ t_p\leq t$
and
$(x-x_q)^2+(y-y_q)^2 \leq  (t_q-t)^2\vmax^2 \ \land \ t\leq t_q$ respectively.
The border of its
\emph{bottom cone} is the set of all points $(t,x,y)$ that satisfy
$$
\Psi_{\Cm{}}(t,x,y,t_p,x_p,y_p,\vmax):=
(x-x_p)^2+(y-y_p)^2 = (t-t_p)^2\vmax^2 \ \land \ t_p\leq t$$ and is
denoted by $\bconep{p}{\vmax}$ or $\bcone{t_p}{x_p}{y_p}{\vmax}$;
and the border of its \emph{upper cone} is the set of all points $(t,x,y)$ that
satisfy  $$ \Psi_{\Cp{}}(t,x,y,t_q,x_q,y_q,\vmax):=
(x-x_q)^2+(y-y_q)^2 = (t_q-t)^2\vmax^2 \ \land \ t\leq t_q$$ and is
denoted by $\uconep{q}{\vmax}$ or $\ucone{t_q}{x_q}{y_q}{\vmax}$.

The set of the two apexes of
$\beadp{p}{q}{\vmax}$ is denotes  $\tau\beadp{p}{q}{\vmax}$, i.e.,
$\tau\beadp{p}{q}{\vmax}=\{p,q\}.$

We call the topological border of the bead $\beadp{p}{q}{\vmax}$ its
\emph{mantel} and denote it by $\partial\beadp{p}{q}{\vmax}$. It can
be easily verified that
 the mantel  consists of the set of points
$(t,x,y)$ that satisfy
$$\displaylines{\quad \Psi_{\partial}(t,x,y,t_p,x_p,y_p,t_q,x_q,y_q,\vmax):= t_p\leq t\leq t_q \land\hfill{}\cr\hfill{}
 \left(2x(x_p-x_q)+x_q^2-x_p^2+2y(y_p-y_q)+y_q^2-y_p^2\leq \vmax^2
2t(t_p-t_q)+t_q^2-t_p^2 \land\right.\hfill{}\cr\hfill{} (x-x_p)^2+(y-y_p)^2 =
(t-t_p)^2\vmax^2\lor (x-x_q)^2+(y-y_q)^2 = (t_q-t)^2\vmax^2 \hfill{}\cr\hfill{}
\land \left.2x(x_p-x_q)+x_q^2-x_p^2+2y(y_p-y_q)+y_q^2-y_p^2\geq
\vmax^2\left(2t(t_p-t_q)+t_q^2-t_p^2\right) \right).}$$

The first conjunction describes the lower half of the mantel and the
second conjunction describes the upper half of the mantel.
The upper and lower half of the mantel are separated by a plane. The
intersection of this plane with  the bead is an ellipse, and the
border of this ellipse is what we will refer to as the \emph{rim} of
the bead. We denote the rim of the bead $\beadp{p}{q}{\vmax}$ by
$\rho\beadp{p}{q}{\vmax}$ and remark that it is described by the
formula
$$\displaylines{\quad
 \Psi_{\rho}(t,x,y,t_p,x_p,y_p,t_q,x_q,y_q,\vmax):= \hfill{}\cr\hfill{}
(x-x_p)^2+(y-y_p)^2 = (t-t_p)^2\vmax^2 \land t_p\leq t\leq t_q
\land \hfill{}\cr\hfill{} 2x(x_p-x_q)+x_q^2-x_p^2+2y(y_p-y_q)+y_q^2-y_p^2 =
\vmax^2\left(2t(t_p-t_q)+t_q^2-t_p^2\right).\quad } $$

The plane in which the rim lies splits the bead into an
\emph{upper-half bead} and a \emph{bottom-half bead}.
 The  \emph{bottom-half bead} is the set of all points $(t,x,y)$ that
satisfy
 $$\displaylines{\Psi_{\B^-}(t,x,y,t_p,x_p,y_p,t_q,x_q,y_q,\vmax):=\hfill{}\cr\hfill{}
(x-x_p)^2+(y-y_p)^2 \leq (t-t_p)^2\vmax^2\land t_p\leq t\leq t_q
\land \hfill{}\cr\hfill{} 2x(x_p-x_q)+x_q^2-x_p^2+2y(y_p-y_q)+y_q^2-y_p^2\leq
\vmax^2\left(2t(t_p-t_q)+t_q^2-t_p^2\right)}$$ and is denoted by
$\bbead{t_p}{x_p}{y_p}{t_q}{x_q}{y_q}{\vmax}$.

The upper bead is the set of all points $(t,x,y)$ that satisfy
 $$\displaylines{\Psi_{\B^+}(t,x,y,t_p,x_p,y_p,t_q,x_q,y_q,\vmax):=\hfill{}\cr\hfill{}
(x-x_q)^2+(y-y_q)^2 \leq (t_q-t)^2\vmax^2\land t_p\leq t\leq t_q
\land \hfill{}\cr\hfill{} 2x(x_p-x_q)+x_q^2-x_p^2+2y(y_p-y_q)+y_q^2-y_p^2 \geq
\vmax^2\left(2t(t_p-t_q)+t_q^2-t_p^2\right)}$$
 and is denoted by
$\ubead{t_p}{x_p}{y_p}{t_q}{x_q}{y_q}{\vmax}$.

\begin{figure}[h]
\centering
\input{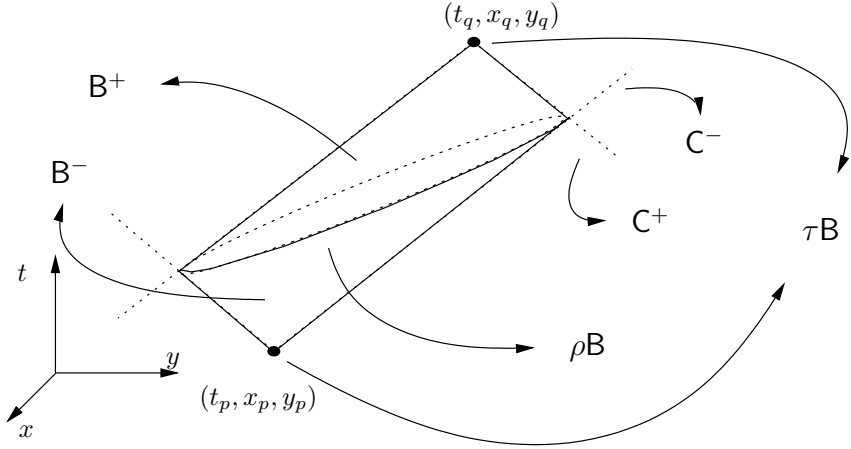}
\caption{A dissection of the bead
$\bead{t_p}{x_p}{y_p}{t_q}{x_q}{y_q}{\vmax}$.}\label{dissection}
\end{figure}

\subsection{The intersection of two cones}\label{subsec:initialcontact}

Let  $\bcone{t_1}{x_1}{y_1}{v_1}$ and
$\bcone{t_2}{x_2}{y_2}{v_2}$ be two bottom cones.
A bottom cone, e.g.,  $\bcone{t_1}{x_1}{y_1}{v_1}$,  can be seen as a circle
in 2-dimensional space $(x,y)$-space with center
$(x_1,y_1)$ and linearly growing radius $(t-t_1)v_1$ as $t_1\leq
t$.

Let us assume that the apex of neither of these
cones is inside the other cone, i.e.,  $(x_1-x_2)^2+(y_1-y_2)^2
> (t_1-t_2)^2v_1^2\lor t_1<t_2$ and
$(x_1-x_2)^2+(y_1-y_2)^2\ab
>(t_1-t_2)^2v_2^2\ab\lor t_2<t_1$.
This assumption implies that at $t_1$ and
$t_2$ neither radius is larger than or equal to the distance between
the two cone centers. So, at first the two circles are disjoint and after growing for some time they intersect in one point.
 We call the first (in time) time-space point  where  the two circles touch
in a single point, and thus for which the sum of the two radii is
equal to the distance between the two centers
the \emph{initial contact} of the
two cones $\bcone{t_1}{x_1}{y_1}{v_1}$
and $\bcone{t_2}{x_2}{y_2}{v_2}$. It is the unique point
$(t,x,y)$ that satisfies the  formula
$$\displaylines{\quad \Psi_{IC^-}(t,x,y,t_1,x_1,y_1,v_1,t_2,x_2,y_2,v_2):= t_1\leq t \wedge t_2 \leq t\land\hfill{} \hfill{}\cr\hfill{}
(x-x_1)^2+(y-y_1)^2 = (t-t_1)^2v_1^2\land
(x-x_2)^2+(y-y_2)^2 = (t-t_2)^2v_2^2\land \hfill{}\cr\hfill{}
( (t-t_1)v_1 + (t-t_2)v_2)^2 = (x_1-x_2)^2+(y_1-y_2)^2.\quad }
$$
The initial contact of two cones $\ucone{t_1}{x_1}{y_1}{v_1}$
and $\ucone{t_2}{x_2}{y_2}{v_2}$ is given by the formula $\Psi_{IC^+}(t,x,y,t_1,x_1,y_1,v_1,t_2,x_2,y_2,v_2)$ that we obtain from $\Psi_{IC^-}$
by replacing in
 $t_1\leq t \wedge t_2 \leq t$ by
$t\leq t_1 \wedge t \leq t_2$.
We denote the singleton sets containing the initial contacts by
$\ICs{\bcone{t_1}{x_1}{y_1}{v_1}}{\bcone{t_2}{x_2}{y_2}{v_2}}$ and
$\ICs{\ucone{t_1}{x_1}{y_1}{v_1}}{\ucone{t_2}{x_2}{y_2}{v_2}}$.

\begin{figure}[h]
\centering
\input{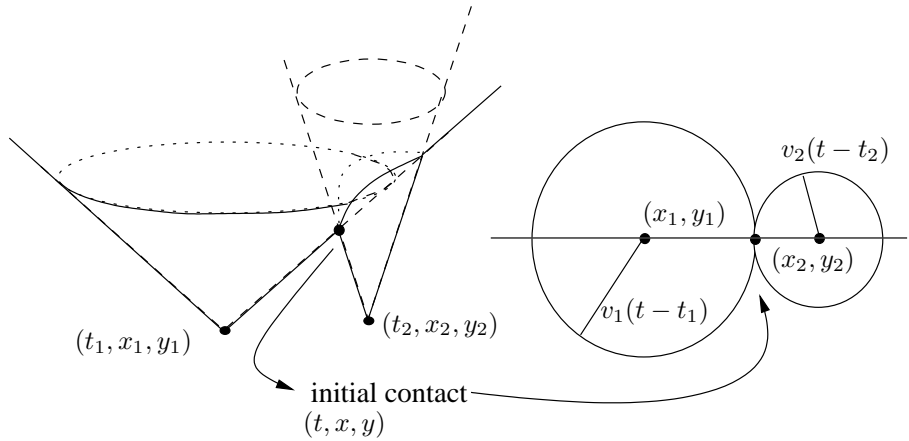}
\caption{Intersecting cones and their initial contact (3-dimensional  view on the left and
2-dimensional view on the right).}\label{intersectingcones}
\end{figure}

From the last equation in of the system in $\Psi_{IC^-}$ and $\Psi_{IC^+}$,
we easily obtain  $t=\frac{\sqrt{(x_1-x_2)^2+(y_1-y_2)^2}+t_1v_1+t_2v_2}{v_1+v_2}.$
To compute the other two coordinates $(x,y)$ of the initial contact, we observe that for in the plane of this time value $t$, it is  on the line segment bounded by $(x_1,y_1)$ and
$(x_2,y_2)$ and that its distance from $(x_1,y_1)$ is $v_1(t-t_1)$ and its
distance from $(x_1,y_1)$ is $v_2(t-t_2)$.  We can conclude that the initial
contact has $(t,x,y)$-coordinates given by the following system of equations
$$\left\{\begin{array}{lll}
 t & = & \frac{\sqrt{(x_1-x_2)^2+(y_1-y_2)^2}+t_1v_1+t_2v_2}{v_1+v_2} \\
 x & = &  x_1+v_1(t-t_1)\frac{x_2-x_1}{\sqrt{(x_2-x_1)^2+(y_2-y_1)^2}} \\
 y & = &  y_1+v_1(t-t_1)\frac{y_2-y_1}{\sqrt{(x_2-x_1)^2+(y_2-y_1)^2}}.
\end{array}\right.$$
This means that we can give more explicit descriptions to replace $\Psi_{IC^-}$ and $\Psi_{IC^+}$.

\section{An analytic solution to the alibi query}\label{sec:2dimalibi}
In this section, we first describe our solution to the alibi query on a geometric level.
Next, we prove its correctness and transform it into an analytic solution and finally we show how to construct a quantifier-free first-order formula out of the analytic solution.

\subsection{Preliminary geometric considerations}\label{subsec:geometricoutline}

\begin{figure}[h]
\centering
\input{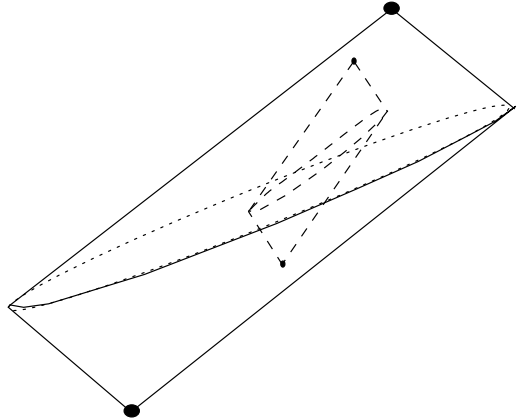}
\caption{One bead is contained in the other.}\label{case1}
\end{figure}

The solution we present is based on the observation that the two main cases of intersection (that do not exclude each other) are: (1)  an apex of one bead is in the other; and (2) the mantels of the beads intersect.

The inclusion of one bead in the other, illustrated in Figure~\ref{case1}, is an example of the first case. It is clear that
if no apex is contained in another bead and we still  assume that the beads intersect, than their mantels must intersect. We show this more formally in Lemma~\ref{mantelcut}. In this second case, the idea is to find a special point (a witness point) that is easily computable and necessarily  in the intersection.

\begin{figure}[h]
\centering
\input{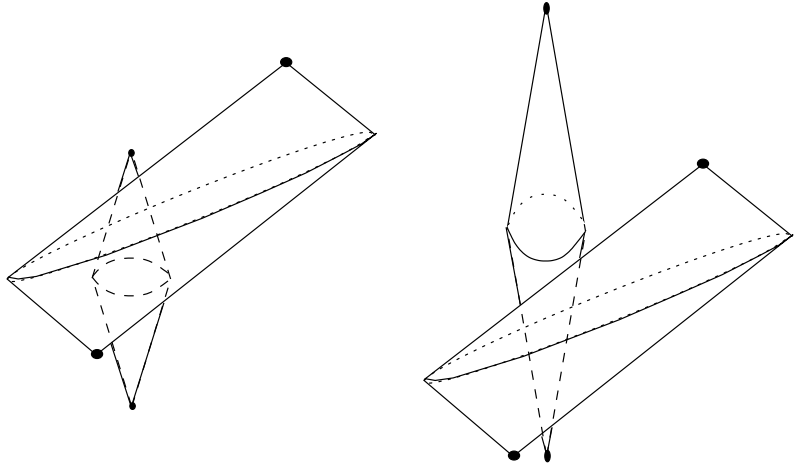}
\caption{Clean cut between cones.}\label{case3}
\end{figure}

Let us consider two beads with
bottom cones $\bcone{t_1}{x_1}{y_1}{v_1}$
and $\bcone{t_2}{x_2}{y_2}{v_2}$ and let us assume that none of the apeces is inside the other cone. One special point is the point of initial contact $\ICs{\bcone{t_1}{x_1}{y_1}{v_1}}{\bcone{t_2}{x_2}{y_2}{v_2}}$. However, this point can not be guaranteed to be in the intersection if the mantels of the two beads intersect, as we will show in the following example. Consider two beads with bottom cones $\bcone{0}{0}{0}{1}$ and $\bcone{0}{2}{0}{1}$.  The intersection is a hyperbola in the plane $x=1$ with equation $t^2-y^2=1$. The initial contact of the two bottom cones is the point $(1,0,1)$.
To show that this point of  initial contact does not need to be in the intersection of the two beads, the idea is to cut this point out of the intersection as follows. Suppose one bead has apexes, $(0,0,0)$ and $(a,b,c)$ and speed $1$. The plane in which its rim lies is given by $-2ax+a^2-2by+b^2+2ct-c^2=0$. This plane cuts the plane $\alpha$ given by the equality $ x=1$ in a line given by the equation  $-2by+2ct-2a+a^2-c^2=0$. Clearly, we can choose $(a,b,c)$ such that the line contains the points $\left(\frac{\sqrt{5}}{2},1,\frac{1}{2}\right)$ and $\left(\sqrt{2},1,1\right)$. Everything below this line will be part of the first bead and the second cone, but the initial contact is situated above the line, effectively cutting it out of the intersection. All this is illustrated in Figure~\ref{iccutout}.

\begin{figure}[h]
\centering
\input{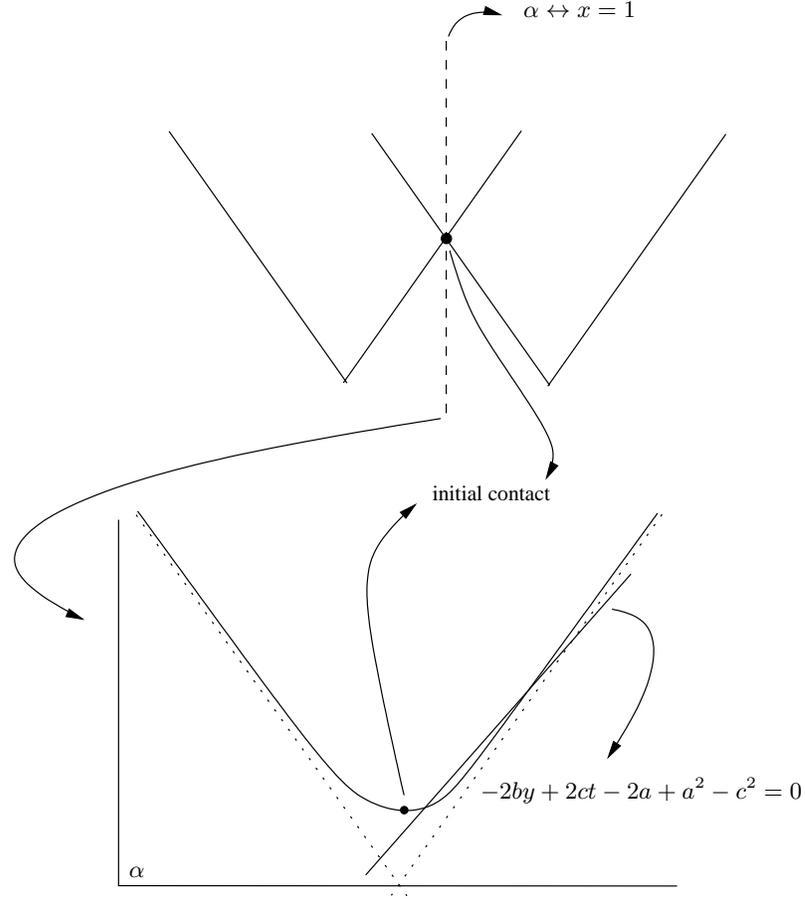}
\caption{The initial contact cut out.}\label{iccutout}
\end{figure}

We notice how the plane in which the rim lies and the rim itself is the evil do-er. If neither rim intersects the mantel of the other bead, then the intersection of mantels is the same as an intersection of cones. In which case the initial contact will not be cut out and can be used to determine if there is intersection in this manner.

Using contraposition on the statement in the previous paragraph we get: if there is an intersection and no initial contact is in the intersection then a rim must intersect the other bead's mantel.

To verify intersection with the apexes and initial contacts is straightforward. Verifying if a rim intersects a mantel results in solving a quartic polynomial equation in one variable and verifying the solution in a single inequality in which no variable appears with a degree higher than one.

\subsection{Outline of the solution}\label{subsec:2dimalibioutline}

Suppose, for the remainder of this section, we wish to verify if the
beads $\B_1=\bead{t_1}{x_1}{y_1}{t_2}{x_2}{y_2}{v_1}$ and
$\B_2=\bead{t_3}{x_3}{y_3}{t_4}{x_4}{y_4}{v_2}$ intersect. Moreover,
we assume the beads are non-empty, i.e., $(x_2-x_1)^2+(y_2-y_1)^2
\leq (t_2-t_1)^2v_1^2$ and $(x_4-x_3)^2+(y_4-y_3)^2 \leq
(t_4-t_3)^2v_2^2$.

We first observe that  an intersection between beads can
be classified into three, mutually exclusive,  cases.
The three cases then are:

\begin{enumerate}
\item[{\bf (I)}] an apex of one bead is contained in the other, i.e.,
$$\tau\B_1 \cap\B_2 \not=\emptyset \mbox{\ or }\B_1 \cap\tau\B_2 \not=\emptyset; $$
\item[{\bf (II)}]  not {\bf (I)}, but the rim of one bead intersects the mantel of the other, i.e.,
$$ \Br_1 \cap \Bm_2 \neq \emptyset \ \mbox{or} \ \Br_2 \cap \Bm_1 \neq \emptyset;$$
\item[{\bf (III)}]  not {\bf (I)} and not {\bf (II)} and the initial contact of the upper or lower cones is in the intersection of the
beads, i.e.,
$$\ICs{\Cm{1}}{\Cm{2}}\subset\B_1\cap\B_2\ \mbox{or}
\  \ICs{\Cp{1}}{\Cp{2}}\subset\B_1\cap\B_2.$$
\end{enumerate}

If none of these three cases occur then the beads do not intersect, as we show in the correctness proof below.
First, we give the following geometric lemma.

\begin{lemma}\label{mantelcut}
If $\B_1\cap\B_2\neq\emptyset$,
$\Ba_1\cap\B_2=\emptyset$ and $\Ba_2\cap\B_1=\emptyset$, then $\Bm_1 \cap \Bm_2 \neq
\emptyset$.\end{lemma}
\begin{proof}
From the assumptions, we
know   there is a point
$p_1$ in $\B_2$, e.g., an apex of $\B_2$,  that is not in $\B_1$. Also, there is a point $p_2$ that is in
$\B_2$ and in $\B_1$. The line segment bounded by $p_1$ and $p_2$
lies in $\B_2$,  since $\B_2$ is convex. The line segment cuts the mantel of $\B_1$
since $p_2$ is inside $\B_1$ and $p_1$ is not. Let $p$ be this point
where the segment bounded by $p_1$ and $p_2$ intersects $\Bm_1$.
This point lies either on the upper-half bead $\uB{1}$ or on the
bottom-half bead $\bB{1}$. Let $r$ be the apex of this half bead.
Since $p$ is inside $\B_2$ and $r$ is not, the line segment bounded
by $p$ and $r$ must cut $\Bm_2$ in a point $q$. This point lies of
course on $\Bm_2$ and on  $\Bm_1$ since the line segment bounded by
$p$ and $r$ is a part of  $\Bm_1$. Hence their mantels must have a
non-empty intersection if the beads have a non-empty intersection and
neither bead contains the apexes of the other. \qed\end{proof}

Now, we show that if $\B_1$ and $\B_2$ intersect and neither $\bf (I)$, nor $\bf (II)$ occur, then $\bf (III)$ occurs.

\begin{theorem}\label{cleancut}
If $\B_1\cap\B_2\neq\emptyset$,
$\tau\B_1 \cap\B_2 =\emptyset$, $\B_1 \cap\tau\B_2  = \emptyset$,
 $ \Br_1 \cap \Bm_2 =\emptyset$ and $\Br_2 \cap \Bm_1 = \emptyset, $
 then   $\ICs{\Cm{1}}{\Cm{2}}\subset\B_1\cap\B_2$ or
$\ICs{\Cp{1}}{\Cp{2}}\subset\B_1\cap\B_2.$
\end{theorem}

\begin{proof}
Let us assume that the hypotheses of the statement of the theorem is
true. It is sufficient to prove that either
$\Cm{1}\cap\Cm{2}\subset\bB{1} \cap \bB{2}$ or
$\Cp{1}\cap\Cp{2}\subset\uB{1}\cap \uB{2}$. We will split the proof
in two cases. From the fourth and fifth hypotheses it follows that
either (1) $\Br_1\subset\B_2$ or $\Br_2\subset\B_1$;  or (2)
$\Br_1\cap\B_2=\emptyset$ and $\Br_2\cap\B_1=\emptyset$.
\medskip

\par \noindent
{\bf Case (1)}: We assume $\Br_2\subset\B_1$ (the case
$\Br_1\subset\B_2$ is completely analogous). We prove
$\Cm{1}\cap\Cm{2}\subset\bB{1}\cap \bB{2}$ (the case for upper cones
is completely analogous). The following argument is  illustrated in Figure~\ref{proof}.

Since $\Br_2\subset\B_1$, we know that $\Br_2$ is inside $\Cm{1}$,
and $(t_3,x_3,y_3)$ is outside. We can show that $v_2<v_1$. Consider
the plane spanned by the two axis of symmetry of both $\Cm{1}$ and
$\Cm{2}$. Both $\Cm{1}$ and $\Cm{2}$ intersect this plane in two
half lines each. Moreover, we know that $\Cm{1}$ intersects the axis
of symmetry of $\Cm{2}$. Let $t_0$ be the moment at which this
happens. Obviously $t_0>t_1$, but we know also know $t_0>t_3$ since
$(t_3,x_3,y_3)$ is  outside $\Cm{1}$. We have that
$v_1(t_0-t_1)=\sqrt{(x_1-x_3)^2+(y_1-y_3)^2}$. Since $\Br_2$ is
inside $\Cm{1}$ and $(t_3,x_3,y_3)$ is outside, this means both half
lines from $\Cm{2}$ intersect the half lines from $\Cm{1}$. Let
$t'_0$ and $t''_0$ be the moments in time at which this happens and
let $t'_0>t''_0$. We have again that $t'_0>t_1$ and $t'_0>t_3$. Then
$v_1(t'_0-t_1)=\sqrt{(x_1-x_3)^2+(y_1-y_3)^2}+v_2(t'_0-t_3)$ if and
only if $v_1(t'_0-t_0)=v_2(t'_0-t_3)$. Since $t_0>t_3$, we get
$v_2<v_1$. This is depicted in Figure~\ref{proof}.

\begin{figure}[h]
\centering
\input{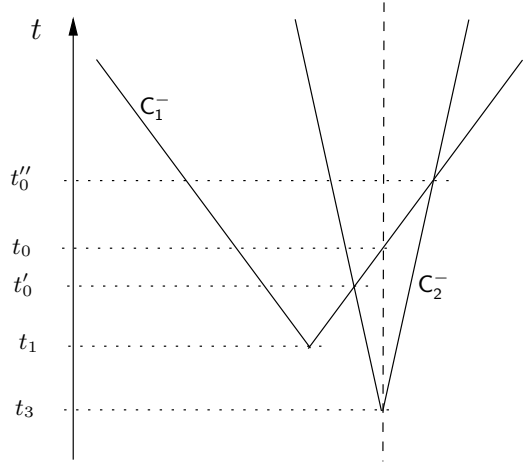}
\caption{Illustration to the proof.}\label{proof}
\end{figure}

It follows that every straight half line starting in $(t_3,x_3,y_3)$
on $\Cm{2}$ intersects $\Cm{1}$ between $(t_3,x_3,y_3)$ and $\Br_2$,
since $\Br_2$ is inside $\Cm{1}$, and $(t_3,x_3,y_3)$ is outside. We
also know that this line does not intersect $\Cm{1}$ beyond $\Br_2$
since the cone $\Cm{2}$ is entirely inside $\Cm{1}$ beyond the rim
$\Br_2$. Therefore, $\Cm{1}\cap\Cm{2}\subset\bB{2}$.

Clearly, $\bB{2}$ intersects $\bB{1}$ since it can not intersect
$\uB{1}$.  We know $\Cm{1}\cap\partial\bB{2}$ is a closed
continuous curve that lies entirely in $\Cm{1}$. This curve is also
contained in $\bB{1}$. Indeed, if we assume this is not the case,
then it intersects the plane in which $\Br_1$ lies, and hence it
intersects $\Br_1$ itself, contradicting  the assumption  $\Br_1
\cap \Bm_2=\emptyset$.
\medskip

\par \noindent
{\bf Case (2)}: Now assume $\Br_1\cap\B_2=\emptyset$ and $\Br_2\cap\B_1=\emptyset$. Clearly,
$v_1$ can not be equal to $v_2$, otherwise the depicted intersection
can not occur.   So suppose without loss of
generality that $v_2<v_1$. Now either $\bB{2}$ intersects both
$\bB{1}$ and $\uB{1}$ or $\uB{2}$ intersects both $\bB{1}$ and
$\uB{1}$. These cases are mutually exclusive because of the
following. If $\uB{2}$ intersects $\uB{1}$ then $\Br_2$ is inside
$\Cp{1}$, likewise if $\bB{2}$ intersects $\bB{1}$ then $\Br_2$ is
inside $\Cm{1}$. Hence $\Br_2\subset \B_1$ which contradicts our
hypothesis. If $\uB{2}$ intersects $\bB{1}$ then $\Br_2$ must be
outside $\Cm{1}$ and thus $\bB{2}$ must be as well, hence $\bB{2}$
intersects neither $\bB{1}$ nor $\uB{1}$. Likewise, if $\bB{2}$
intersects $\uB{1}$ then $\uB{2}$ can not intersect $\B_1$.

To prove that if $\bB{2}$ intersects $\bB{1}$ then it also
intersects $\uB{1}$ and if $\bB{2}$ intersects $\uB{1}$ then it also
intersects $\bB{1}$ we proceed as follows (the case for $\uB{2}$ is
analogous). Suppose $\bB{2}$ intersects $\bB{1}$, then
$\bB{2}\cap\bB{1}\subset\B_1$, but $\Br_2$ is outside $\B_1$, that
means $\bB{2}$ must intersect $\uB{1}$ since it can not intersect
$\bB{1}$ anymore. This is the ``what goes in must come
out''-principle. Likewise, suppose $\bB{2}$ intersects $\uB{1}$, then
$\bB{2}\cap\uB{1}\subset\B_1$, but $(t_3,x_3,y_3)$ is outside
$\B_1$, that means $\bB{2}$ must intersect $\bB{1}$ since it can not
intersect $\bB{1}$ anymore.

So suppose now that $\bB{2}$ intersects both $\bB{1}$ and $\uB{1}$ (the case for $\uB{2}$ is completely analogous). If $\bB{2}$
intersects $\bB{1}$ that means $\Br_2$ is completely inside $\Cm{1}$
and therefore that $\Cm{1}\cap \Cm{2} \subset \bB{2}$. We proceed
like in the first case, we know that $\Cm{1}\cap\bB{2}$ is a closed
continuous curve. This curve lies entirely in $\Cm{1}$. If this
curve is not entirely in $\bB{1}$ that means it intersects the plane
in which $\Br_1$ lies, and hence intersects $\Br_1$ itself. But this
is contradictory to the assumption that $\Br_1 \cap
\Bm_2=\emptyset$.  \qed\end{proof}

\begin{figure}[h]
\centering
\epsfig{file=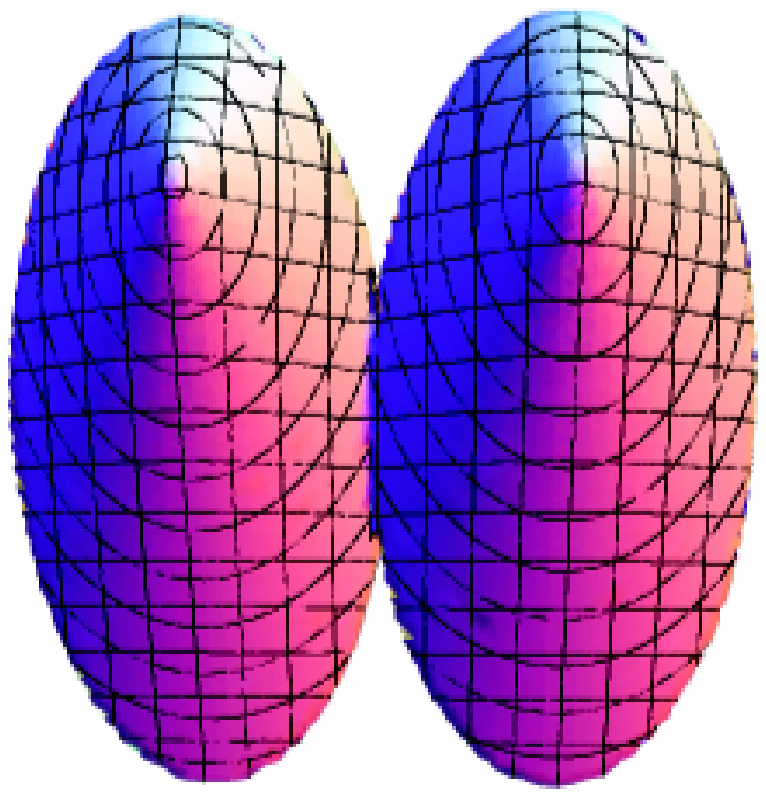, height=3in}
\\
\epsfig{file=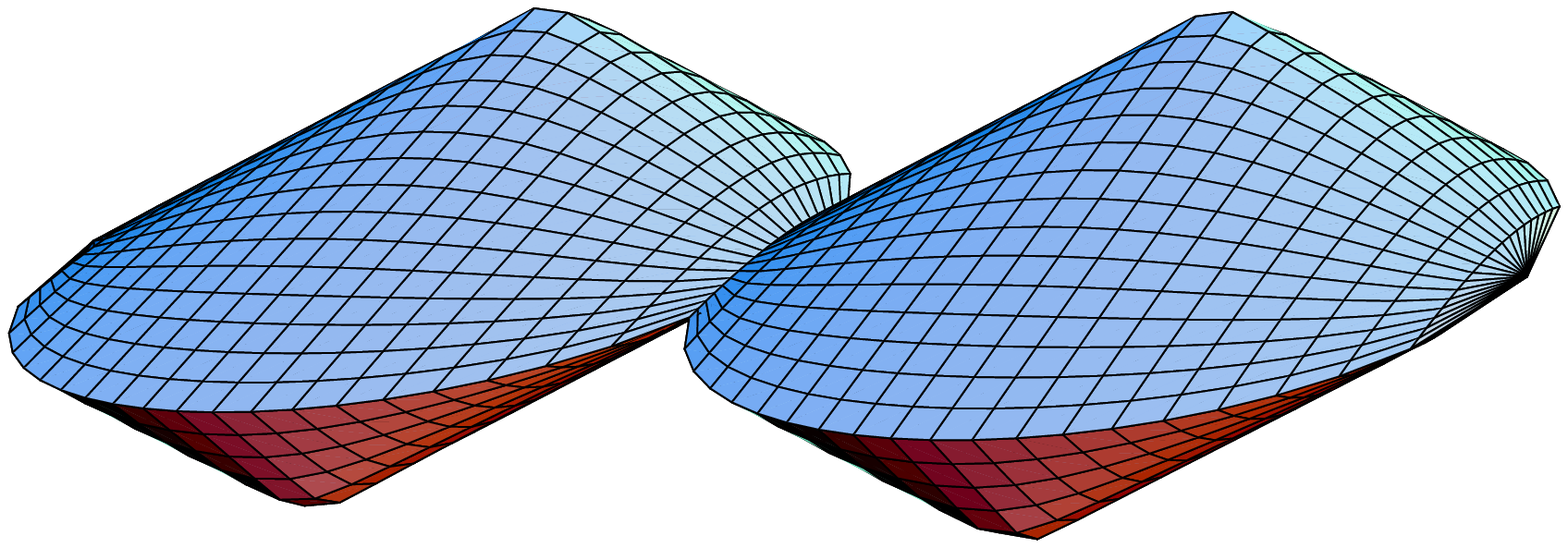, width=5in, height=2in}
\caption{Case \textbf{(II)} is not redundant: $\bead{0}{0}{0}{2}{0}{2}{1.9}$ and $\bead{0}{3}{0}{2}{3}{2}{1.9}$ seen from the top and the side.}\label{counterex}
\end{figure}

In Theorem~\ref{cleancut}, we proved that if there is an intersection
and neither rim cuts the other bead's mantel and neither apex of a
bead is contained in the other then there must be an initial contact
in the intersection. Visualizing how beads intersect might tempt one
to think there is always an initial contact in the intersection.
There exist counterexamples in which there is an intersection and no initial contact is in that intersection.
That means case \textbf{(II)} is not redundant. This situation is depicted in Figure~\ref{counterex}.
The beads are $\bead{0}{0}{0}{2}{0}{2}{1.9}$ and $\bead{0}{3}{0}{2}{3}{2}{1.9}$.

\begin{figure}[h]
\centering
\epsfig{file=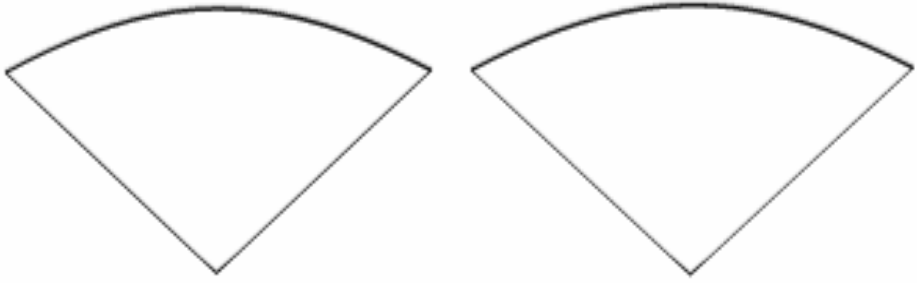, width=5in, height=1.5in}
\caption{Intersection of Figure~\ref{counterex} with the plane $y=0$.}\label{counterex1}
\end{figure}

It is clear that the initial contact of the bottom cones lies in the plane spanned by the axis of symmetry of those bottom cones, in this case this is the plane $y=0$. The intersection of Figure~\ref{counterex} can be seen in Figure~\ref{counterex1}, where the two beads clearly have no intersection and thus no initial contact in the intersection.

\begin{figure}[h]
\centering
\epsfig{file=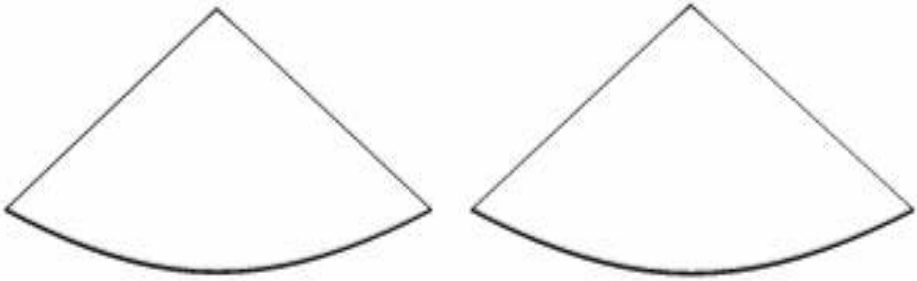, width=5in, height=1.5in}
\caption{Intersection of Figure~\ref{counterex} with the plane $y=3$.}\label{counterex2}
\end{figure}

In the case of the upper cones the initial contact must lie in the plane $y=3$. The intersection of Figure~\ref{counterex} can be seen in Figure~\ref{counterex2}, where the two beads clearly have no intersection and there is again no initial contact in the intersection.

 This concludes the outline.

\subsection{A formula for Case (I)}\label{subsec:caseI}

In Case \textbf{(I)}, we verify whether $\tau\B_1 \cap\B_2 \not=\emptyset$ or $\B_1
\cap\tau\B_2 \not=\emptyset$. To check if that is the case we merely
need to verify if one of the apexes satisfies the set of equations
of the other bead. In this way we obtain
$$\displaylines{\Phi_{\textbf{I}}\left(t_1,x_1,y_1,t_2,x_2,y_2,v_1,t_3,x_3,y_3,t_4,x_4,y_4,v_2\right)
:=\hfill{}\cr \hfill{}\left(
\Psi_{\B}\left(t_3,x_3,y_3,t_1,x_1,y_1,t_2,x_2,y_2,v_1\right)\lor
\Psi_{\B}\left(t_4,x_4,y_4,t_1,x_1,y_1,t_2,x_2,y_2,v_1\right)\right.\lor\cr\hfill{}
\left.\Psi_{\B}\left(t_1,x_1,y_1,t_3,x_3,y_3,t_4,x_4,y_4,v_2\right)\lor
\Psi_{\B}\left(t_2,x_2,y_2,t_3,x_3,y_3,t_4,x_4,y_4,v_2\right)\right).}$$

For the following sections we assume that the apex sets of the beads
are not singletons, i.e., $t_1<t_2$ and $t_3<t_4$.

\subsection{A formula for Case (II)}\label{subsec:caseII}

Now, let us assume that $\Phi_{\textbf{I}}$ failed in the previous section.
Note that we can always apply a speed-preserving~\cite{ons-icdt}
transformation to $\R\times\R^2$ to obtain easier coordinates. We
can always find a transformation such that
$(t'_1,x'_1,y'_1)=(0,0,0)$ and that the line-segment connecting
$(t'_1,x'_1,y'_1)$ and $(t'_2,x'_2,y'_2)$ is perpendicular to the
$y$-axis, i.e., $y'_2=0$.  This transformation is a composition of a
translation in $\R\times\R^2$, a spatial rotation in $\R^2$ and a
scaling in $\R\times\R^2$~\cite{ons-icdt}. Let the coordinates
without a prime be the original set, and let coordinates with a
prime be the image of the same coordinates without a prime under
this transformation. Note that we do not need to transform back
because the query is invariant under such
transformations~\cite{ons-icdt}. The following formula returns the
transformed coordinates $(t',x',y')$ of $(t,x,y)$ given the points
$(t_1,x_1,y_1)$ and $(t_2,x_2,y_2)$:
$$\displaylines{\varphi_{A}(t_1,x_1,y_1,t_2,x_2,y_2,t,x,y,t',x',y'):=\left(y_2\neq y_1\ \land
\right.\hfill{}\cr\hfill{}
 \ t'=(t-t_1)\sqrt{(x_2-x_1)^2+(y_2-y_1)^2}\
\land \ x'=(x-x_1)(x_2-x_1)\cr\hfill{}\left.+(y-y_1)(y_2-y_1)\ \land
\ y'=(x-x_1)(y_1-y_2)+(y-y_1)(x_2-x_1)\right) \cr\hfill{}\lor\
\left(y_2= y_1\ \land \ t'=(t-t_1)\ \land \ x'=(x-x_1)\ \land \
y'=(y-y_1)\right). }$$

The translation is over the vector $(-t_1,-x_1,-y_1)$, the rotation
over minus the angle that $(t_2-t_1,x_2-x_1,y_2-y_1)$ makes with the
$x$-axis, and a scaling by a factor
$\sqrt{(x_2-x_1)^2+(y_2-y_1)^2}$. Notice that the rotation and
scaling only need to occur if $y_2$ is not already in place, i.e., if
$y_2\neq y_1$.

The formula $\psi_{crd}(t'_1,\ab x'_1,\ab y'_1,\ab t_1,\ab x_1,\ab
y_1,\ab t'_2,\ab x'_2,\ab y'_2,\ab t_2,\ab x_2,\ab y_2,\ab t'_3,\ab
x'_3,\ab y'_3,\ab t_3,\ab x_3,\ab y_3,\ab t'_4,\ab x'_4,\ab y'_4,\ab
t_4,\ab x_4,\ab y_4)$ is short for $ \varphi_{A}(t_1,\ab x_1,\ab
y_1,\ab t_2,\ab x_2,\ab y_2,\ab t_1,\ab x_1,\ab y_1,\ab t'_1,\ab
x'_1,\ab y'_1)\ \land \ \varphi_{A}(t_1,\ab x_1,\ab y_1,\ab t_2,\ab
x_2,\ab y_2,\ab t_2,\ab x_2,\ab y_2,\ab t'_2,\ab x'_2,\ab y'_2)\
\land \ \varphi_{A}(t_1,\ab x_1,\ab y_1,\ab t_2,\ab x_2,\ab y_2,\ab
t_3,\ab x_3,\ab y_3,\ab t'_3,\ab x'_3,\ab y'_3)\ \land \
\varphi_{A}(t_1,\ab x_1,\ab y_1,\ab t_2,\ab x_2,\ab y_2,\ab t_4,\ab
x_4,\ab y_4,\ab t'_4,\ab x'_4,\ab y'_4)$.

 This transformation yields some
simple equations for the rim $\Br_1$:
$$\Br_1\leftrightarrow\left\{
\begin{array}{l}
x^2+y^2 = t^2v_1^2\\
2x(-x'_2)+ x'^2_2 = v_1^2(2t(-t'_2)+t'^2_2)\\
0\leq t\leq t'_2\ .
\end{array}
\right. $$

Not only that, but with these equations we can deduce a simple
parametrization in the $x$-coordinate for the rim,
$$\Br_1\leftrightarrow\left\{
\begin{array}{l}
t=\frac{2xx'_2-x'^2_2+v_1^2t'^2_2}{2v_1^2t'_2}\\
y=\pm\sqrt{v_1^2\left(\frac{2xx'_2-x'^2_2+v_1^2t'^2_2}{2v_1^2t'_2}\right)^2-x^2}\\
0\leq t\leq t'_2\ .\end{array} \right.$$

We remark that this implies $t'_2\neq 0$ and $v_1\neq 0$.  If $t'_2 = 0$,
then $\B_1$ is a point, hence degenerate. If $v_1 = 0$, then $\B_1$
is a line segment, and again degenerate. Next we will inject these
parameterizations in the constraints for $\partial\uB{2}$ and
$\partial\bB{2}$ separately. The constraints for $\partial\bB{2}$
are
$$\left\{
\begin{array}{l}
(x-x'_3)^2+(y-y'_3)^2 = (t-t'_3)^2v_2^2\\
2x(x'_3-x'_4)+x'^2_4-x'^2_3+2y(y'_3-y'_4)+y'^2_4-y'^2_3\leq v_2^2\left(2t(t'_3-t'_4)+t'^2_4-t'^2_3\right)\\
t'_3\leq t\leq t'_4 \ .
\end{array}
\right. $$

We will explain how to proceed to compute the intersection with
$\partial\bB{2}$ and simply reuse formulas for intersection with
$\partial\uB{2}$. First, we insert our expressions for $x$ and $y$ in
the first equation. This is equivalent to computing intersections of
$\Br_1$ with $\Cm{2}$ and gives
$$\displaylines{\left(x-x'_3\right)^2+
\left(\pm\sqrt{v_1^2\left(\frac{2xx'_2-x'^2_2+v_1^2t'^2_2}{2v_1^2t'_2}\right)^2-x^2}-y'_3\right)^2
\cr =
\left(\frac{2xx'_2-x'^2_2+v_1^2t'^2_2}{2v_1^2t'_2}-t'_3\right)^2v_2^2,}$$

or equivalently

$$\displaylines{
\pm2y'_3\sqrt{v_1^2\left(2xx'_2-x'^2_2+v_1^2t'^2_2\right)^2-\left(2v_1^2t'_2\right)^2x^2}
 =\cr
\left(2xx'_2-x'^2_2+v_1^2t'^2_2-\left(2v_1^2t'_2\right)t'_3\right)^2v_2^2-\left(2v_1^2t'_2\right)^2(x-x'_3)^2\cr
-\left(2v_1^2t'_2\right)^2y'^2_3-
\left(v_1^2\left(2xx'_2-x'^2_2+v_1^2t'^2_2\right)^2-\left(2v_1^2t'_2\right)^2x^2\right)}$$

or equivalently

$$\displaylines{
\pm v_12y'_3\sqrt{x^24\left(x'^2_2-v_1^2t'^2_2\right)
+x4x'^2_2\left(v_1^2t'^2_2-x'^2_2\right)+\left(v_1^2t'^2_2-x'^2_2\right)^2}
 \cr=x^24x'^2_2\left(v_2^2-v_1^2\right)+
 \cr x4\left(-x'^2_2v_2^2\left(-x'^2_2+v_1^2t'^2_2-4v_1^4t'^2_2t'_3\right)
 +2v_1^4t'^2_2x'_3+v_1^2x'_2\left(v_1^2t'^2_2-x'^2_2\right)\right)\cr+
 \left(v_2^2\left(-x'^2_2+v_1^2t'^2_2-4v_1^4t'^2_2t'_3\right)^2-4v_1^4t'^2_2\left(x'^2_3+y'^2_3\right)
 -v_1^4\left(-x'^2_2+v_1^2t'^2_2\right)\right)\ .
}$$

By squaring left and right hand in this last expression, we rid ourselves
of the square root and obtain the following polynomial equation of
degree four. Squaring may create new solutions, so to ensure we only get useful solutions, we have to add the condition that the square root exists. This is the case if and only if
$$\phi_{}\sqrt{\ }(x,t'_2,x'_2,v_1):=x^24\left(x'^2_2-v_1^2t'^2_2\right)
+x4x'^2_2\left(v_1^2t'^2_2-x'^2_2\right)+\left(v_1^2t'^2_2-x'^2_2\right)^2\geq 0$$
is satisfied.

We notice that if $\B_1$ is degenerate, i.e., $x'^2_2=v_1^2t'^2_2$, then the square root vanishes and the polynomial in
$\phi_4$ is the square of a polynomial of degree two, yielding to at
most two roots and intersection points as we expect. The case were $v_1=0$ is captured by the formula in the next section, that is why we leave that case out here and demand that $v_1\neq 0$. So the
following still works if one or both beads is degenerate:
$$\displaylines{\phi_4(x,t'_2,x'_2,v_1,t'_3,x'_3,y'_3,v_2):=\exists
a\exists b\exists c\exists d\exists e \left( ax^4+bx^3+cx^2+dx+e=0
\right.
\hfill{}\cr\hfill{} \land \
a=\left(4x'^2_2\left(v_2^2-v_1^2\right)\right)^2 \ \land \
b=-32x'^4_2v_2^2\left(v_2^2-v_1^2\right)\left(-x'^2_2+v_1^2t'^2_2-4v_1^4t'^2_2t'_3
 \right.\cr\hfill{}\left.+2v_1^4t'^2_2x'_3+v_1^2x'_2\left(v_1^2t'^2_2-x'^2_2\right)\right)\
 \land \ c=8\left(x'^2_2-v_1^2t'^2_2\right)
\left(-4v_1^4t'^2_2\left(x'^2_3+y'^2_3\right)+
\right.\cr\hfill{}\left.+v_2^2\left(-x'^2_2+v_1^2t'^2_2-4v_1^4t'^2_2t'_3\right)^2
 -v_1^4\left(-x'^2_2+v_1^2t'^2_2\right)\right)+
 \left(2v_1y'_3\right)^2\left(x'^2_2-v_1^2t'^2_2\right)
 \cr\hfill{}+
 \left(4\left(-x'^2_2v_2^2\left(-x'^2_2+v_1^2t'^2_2-4v_1^4t'^2_2t'_3\right)
 +2v_1^4t'^2_2x'_3+v_1^2x'_2\left(v_1^2t'^2_2-x'^2_2\right)\right)\right)^2
\cr\hfill{} \land \
d=8\left(-x'^2_2v_2^2\left(-x'^2_2+v_1^2t'^2_2-4v_1^4t'^2_2t'_3\right)
 +2v_1^4t'^2_2x'_3+v_1^2x'_2\left(v_1^2t'^2_2-x'^2_2\right)\right)\cr\hfill{}
 \left(v_2^2\left(-x'^2_2+v_1^2t'^2_2-4v_1^4t'^2_2t'_3\right)^2-4v_1^4t'^2_2\left(x'^2_3+y'^2_3\right)
 -v_1^4\left(-x'^2_2+v_1^2t'^2_2\right)\right)\cr\hfill{}+
\left(2v_1y'_3\right)^2\left(4x'^2_2\left(v_1^2t'^2_2-x'^2_2\right)\right)
\ \land \ e=
\left(2v_1y'_3\right)^2\left(v_1^2t'^2_2-x'^2_2\right)^2+\cr\hfill{}
\left.
\left(v_2^2\left(-x'^2_2+v_1^2t'^2_2-4v_1^4t'^2_2t'_3\right)^2-4v_1^4t'^2_2\left(x'^2_3+y'^2_3\right)
 -v_1^4\left(-x'^2_2+v_1^2t'^2_2\right)\right)^2\right) .}$$
The quantifiers we introduced here are only in place for esthetical
considerations and can be eliminated by direct substitution.

We note that if $v_1=v_2$, we get polynomials of degree   merely two. This can be solved
in an exact manner using nested square roots (or Maple if you will).
This gives us at most four values for $x$. Let
$$\phi_{roots}(x_a,x_b,x_c,x_d,t'_2,x'_2,v_1,t'_3,x'_3,y'_3,v_2)$$
be a formula that returns all four real roots, if they exist, that
satisfy both $\phi_4(x,t'_2,x'_2,\ab v_1,t'_3,x'_3,y'_3,v_2)$ and $\phi_{}\sqrt{\ }(x,t'_2,x'_2,v_1)$.
We
substitute these values in the parameter equations of $\Br_1$. By
substituting these in the last equation above, we can determine the
sign of the square root we need to take for $y$. A point $(t,x,y)$
satisfies the following formula is a point on $\Br_1$, but instead
of using the square root for $y$, we use an expression from above to
get the correct sign for the square root if $y'_3\neq 0$. If $y'_3=0$ we have to use the square root expression and then it does not matter which sign the square root has; we need both:
$$\displaylines{\psi_{\rho}(t,x,y,t'_2,x'_2,v_1,t'_3,x'_3,y'_3,v_2):=\left(y'_3\neq0\ \land \
t\left(2v_1^2t'_2\right)=2xx'_2-x'^2_2+v_1^2t'^2_2\
\land\right.\hfill{}\cr\hfill{} 2y'_3\left(2v_1^2t'_2\right)y=
\left(2xx'_2-x'^2_2+v_1^2t'^2_2-\left(2v_1^2t'_2\right)t'_3\right)^2v_2^2-\left(2v_1^2t'_2\right)^2(x-x'_3)^2\cr\hfill{}
\left.-\left(2v_1^2t'_2\right)^2y'^2_3-
\left(v_1^2\left(2xx'_2-x'^2_2+v_1^2t'^2_2\right)^2-\left(2v_1^2t'_2\right)^2x^2\right)\
\land \ \ 0\leq t\leq t'_2\ \right)\cr\hfill{}\lor
\left(y'_3=0\ \land \
t\left(2v_1^2t'_2\right)=2xx'_2-x'^2_2+v_1^2t'^2_2\ \land \ 0\leq t\leq t'_2\
\land\right.\cr\hfill{} \left.\left(2v_1^2t'_2\right)^2y^2=
\left(2xx'_2-x'^2_2+v_1^2t'^2_2\right)^2-\left(2v_1^2t'_2\right)^2x^2
\right).}$$

The four roots give us four spatio-temporal points on
$\Br_1\cap\Cm{2}$. In order for these points $(t,x,y)$ to be in
$\Br_1\cap\partial\bB{2}$, they need to satisfy
$$\displaylines{\psi_-(t,x,y,t'_3,x'_3,y'_3,t'_4,x'_4,y'_4,v_2):=
\hfill{}\cr\hfill{}
2x(x'_3-x'_4)+x'^2_4-x'^2_3+2y(y'_3-y'_4)+y'^2_4-y'^2_3\leq
v_2^2\left(2t(t'_3-t'_4)+t'^2_4-t'^2_3\right)\ .}$$ This formula
returns {\sc True} if $(t,x,y)$ lies in the same half space as the
bottom-half bead.

\begin{figure}[h]
\centering
\input{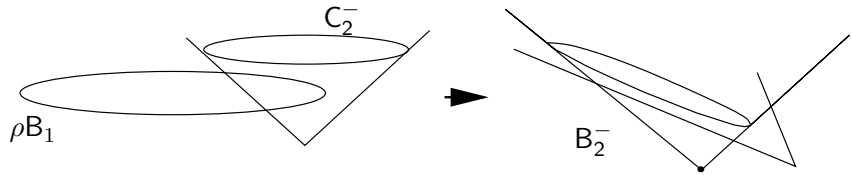}
\caption{The rim intersects the cone and solutions are verified in a half-space.}\label{rimconehalfspace}
\end{figure}

 The formula $\psi_+$ returns {\sc True}  if $(t,x,y)$ lies in the
same half space as the upper-half bead, i.e.,
$\psi_+(t,x,y,t'_3,x'_3,y'_3,t'_4,x'_4,y'_4,v_2):=\psi_-(t,\ab x,\ab
y,\ab t'_4,\ab x'_4,\ab y'_4,\ab t'_3,\ab x'_3,\ab y'_3,\ab v_2)$.
By combining $\psi_{\rho}(t,x,y,t'_2,x'_2,v_1,t'_3,x'_3,y'_3,v_2)$
and $\psi_-(t,\ab x,\ab y,\ab \hat{t},\ab \hat{x},\ab \hat{y},\ab
\tilde{t},\ab \tilde{x},\ab \tilde{y},\ab v)$ we get a formula that
decides the emptyness of the intersection $\Br_1\cap\partial\bB{2}$ in terms of  a parameter $x$:
$$\displaylines{\psi_{\rho\cap\partial^\pm}(x,t'_2,x'_2,v_1,t'_3,x'_3,y'_3,v_2,\hat{t},\hat{x},\hat{y},
\tilde{t},\tilde{x},\tilde{y},v):=\hfill{}\cr\hfill{}\exists y  \
\left(y'_3=0\ \land \  y^2\left(2v_1^2t'_2\right)^2=\left(2xx'_2-x'^2_2+v_1^2t'^2_2\right)^2v_1^2-\left(2v_1^2t'_2\right)^2x^2\right)\ \lor \
\cr\hfill{}
\left(y'_3\neq0\ \land \ 2y'_3\left(2v_1^2t'_2\right)y=
\left(2xx'_2-x'^2_2+v_1^2t'^2_2-\left(2v_1^2t'_2\right)t'_3\right)^2v_2^2-\left(2v_1^2t'_2\right)^2\right.
\cr\hfill{}
(x-x'_3)^2\left.-\left(2v_1^2t'_2\right)^2y'^2_3-
\left(v_1^2\left(2xx'_2-x'^2_2+v_1^2t'^2_2\right)^2-\left(2v_1^2t'_2\right)^2x^2\right)\right)
\land\ \cr\hfill{}
\left(2x(\hat{x}-\tilde{x})+\tilde{x}^2-\hat{x}^2+2y
(\hat{y}-\tilde{y})
+\tilde{y}^2-\hat{y}^2\right)\left(2v_1^2t'_2\right)
\leq\cr\hfill{}
v^2\left(2\left(2v_1^2t'_2\right)\left(2xx'_2-x'^2_2+v_1^2t'^2_2\right)
(\hat{t}-\tilde{t})+\left(2v_1^2t'_2\right)\left(\tilde{t}^2-\hat{t}^2\right)\right)\
\land \ 0\leq t'_2\cr\hfill{}
\left(2xx'_2-x'^2_2+v_1^2t'^2_2\right)\leq 2v_1^2t'^3_2\ \land \
\left(\hat{t}\left(2v_1^2t'^2_2\right)\leq
t'_2\left(2xx'_2-x'^2_2+v_1^2t'^2_2\right)\leq
\tilde{t}\left(2v_1^2t'^2_2\right)\right. \cr\hfill{}\lor\
\left.\tilde{t}\left(2v_1^2t'^2_2\right)\leq
t'_2\left(2xx'_2-x'^2_2+v_1^2t'^2_2\right)\leq
\hat{t}\left(2v_1^2t'^2_2\right)\right)\ .}$$

We are ready now to construct the formula that decides if $\Br_1$
and $\bB{2}$ have a non-empty intersection:
$$\displaylines{\varphi_{\rho_1\cap\partial_2^-}(t'_2,x'_2,v_1,t'_3,x'_3,y'_3,t'_4,x'_4,y'_4,v_2):=
\exists x \exists x_a\exists x_b\exists x_c\exists x_d\ \left(\right.
\hfill{}\cr
\phi_{roots}(x_a,x_b,x_c,x_d,t'_2,x'_2,v_1,t'_3,x'_3,y'_3,v_2)\
\land \  \left(x=x_a\lor x=x_b\lor x=x_c\lor
x=x_d\right)\cr\hfill{}
\left.
\land \
\psi_{\rho\cap\partial^\pm}(x,t'_2,x'_2,v_1,t'_3,x'_3,y'_3,v_2,t'_3,x'_3,y'_3,t'_4,x'_4,y'_4,v'_2)\
\right)
.}$$

The formula that decides if $\Br_1$ intersects $\partial\uB{2}$
looks strikingly similar:
$$\displaylines{\varphi_{\rho_1\cap\partial_2^+}(t'_2,x'_2,v_1,t'_3,x'_3,y'_3,t'_4,x'_4,y'_4,v_2):=
\exists x \exists x_a\exists x_b\exists x_c\exists x_d\
\left(
\right.
\hfill{}\cr
\phi_{roots}(x_a,x_b,x_c,x_d,t'_2,x'_2,v_1,t'_4,x'_4,y'_4,v_2)\
\land \  \left(x=x_a\lor x=x_b\lor x=x_c\lor
x=x_d\right)\cr\hfill{}
\left.
\land \
\psi_{\rho\cap\partial^\pm}(x,t'_2,x'_2,v_1,t'_4,x'_4,y'_4,v_2,t'_4,x'_4,y'_4,t'_3,x'_3,y'_3,v'_2)\
\right)
.}$$

The quantifiers introduced here can also be eliminated in a
straightforward manner. Notice that $\phi_{roots}$ acts as a
function rather than a formula that inputs
$(t'_2,x'_2,v_1,t'_4,x'_4,y'_4,v_2)$ to construct a polynomial of
degree four and returns the four roots $(x_a,x_b,x_c,x_d)$, if they
exist, of that polynomial. The existential quantifier for the
variable $x$ is used to cycle through those roots to see if any of
them does the trick. Finally we are ready to present the formula for
Case \textbf{(II)}:
$$\displaylines{\Phi_{\textbf{II}}\left(t_1,x_1,y_1,t_2,x_2,y_2,v_1,t_3,x_3,y_3,t_4,x_4,y_4,v_2\right)
:=\left(v_1\neq0\ \land \ v_2 \neq  0\right)
\ \land\cr \lnot \
\Phi_{\textbf{I}}\left(t_1,x_1,y_1,t_2,x_2,y_2,v_1,t_3,x_3,y_3,t_4,x_4,y_4,v_2\right)
\land \cr \exists t'_1\exists x'_1\exists y'_1\exists t'_2\exists
x'_2\exists y'_2\exists t'_3\exists x'_3\exists y'_3\exists
t'_4\exists x'_4\exists  y'_4
\left(
\right.
\cr
\psi_{crd}(t'_1,x'_1,y'_1,t_1,x_1,y_1,t'_2,x'_2,y'_2,t_2,x_2,y_2,t'_3,x'_3,y'_3,t_3,x_3,y_3,t'_4,x'_4,y'_4,t_4,x_4,y_4)
\cr\land\
\left(\varphi_{\rho_1\cap\partial_2^-}(t'_2,x'_2,v_1,t'_3,x'_3,y'_3,t'_4,x'_4,y'_4,v'_2)\
\lor \right.\hfill{}\cr\hfill{}\left.
\varphi_{\rho_1\cap\partial_2^+}(t'_2,x'_2,v_1,t'_3,x'_3,y'_3,t'_4,x'_4,y'_4,v_2)\right)\
\lor\cr
\psi_{crd}(t'_3,x'_3,y'_3,t_3,x_3,y_3,t'_4,x'_4,y'_4,t_4,x_4,y_4,t'_1,x'_1,y'_1,t_1,x_1,y_1,t'_2,x'_2,y'_2,t_2,x_2,y_2)
\cr
\land\
\left(\varphi_{\rho_1\cap\partial_2^-}(t'_3,x'_3,v_2,t'_1,x'_1,y'_1,t'_2,x'_2,y'_2,v'_1)\
\lor \right.\hfill{}\cr\hfill{}
\left. \left.
\varphi_{\rho_1\cap\partial_2^+}(t'_3,x'_3,v_2,t'_1,x'_1,y'_1,t'_2,x'_2,y'_2,v'_1)\right)\right)
\
.}$$

The reader may notice that a lot of quantifiers have been introduced
in the formula above. These quantifiers are merely there to
introduce easier coordinates and can be straightforwardly computed
(and eliminated) by the formula $\psi_{crd}$ and hence the formula
$\varphi_{A}(t_1,x_1,y_1,t_2,x_2,y_2,t,x,y,t',x',y')$. The latter
actually acts like a function, parameterized by
$(t_1,x_1,y_1,t_2,x_2,y_2)$, that inputs $(t,x,y)$ and outputs
$(t',x',y')$.

\subsection{A formula for Case (III)}\label{subsec:caseIII}

Here, we assume that both $\varphi_{\textbf{I}}$ and
$\varphi_{\textbf{II}}$ fail. So, there is no apex contained in the
other bead and neither rim cuts the mantel of the other bead.

As we proved in Theorem~\ref{cleancut}, the intersection between two
half beads will reduce to the intersection between two cones and
that means there is an initial contact that is part of the
intersection. To verify if this is the case we compute the two
initial contacts and verify if they are effectively part of the
intersection.

Using the expression for the initial contact $\IC(\Cm{1},\Cm{2})$, we
computed in Section~\ref{subsec:initialcontact} we can construct a formula
that decides if it is part of $\bB{1}\cap\bB{2}$. We will recycle
the formulas $\psi_-$ from the previous section to construct an
expression without the need for extra variables. The following
formula that returns {\sc True} if $\IC(\Cm{1},\Cm{2})=(t_0,x_0,y_0)$
satisfies $\psi_-(t_0,x_0,y_0,t',x',y',\hat{t},\hat{x},\ab
\hat{y},v)$:
$$\displaylines{\phi_-(t_1,x_1,y_1,v_1,t_3,x_3,y_3,v_2,t',x',y',\hat{t},\hat{x},\hat{y},v):=\hfill{}\cr\hfill{}
2(x'-\hat{x})\left((x_1v_2+x_3v_1)\sqrt{(x_1-x_3)^2+(y_1-y_3)^2}+
v_1\left((t_3-t_1)v_2\right)(x_3-x_1)\right)\cr\hfill{}
+2(y'-\hat{y})\left((y_1v_2+y_3v_1)\sqrt{(x_1-x_3)^2+(y_1-y_3)^2}+
v_1\left((t_3-t_1)v_2\right)(y_3-y_1)\right)\cr\hfill{}
+\sqrt{(x_1-x_3)^2+(y_1-y_3)^2}(v_1+v_2)\left(\hat{x}^2-x'^2+\hat{y}^2-y'^2\right)\leq
v^2\left(\left(\hat{t}^2-t'^2\right)(v_1+v_2)\right.\cr\hfill{}
+\left.2\left(\sqrt{(x_1-x_3)^2+(y_1-y_3)^2}+t_1v_1+t_3v_2\right)
\left(t'-\hat{t}\right)\right)\sqrt{(x_1-x_3)^2+(y_1-y_3)^2}\ .}$$

The following formula expresses that the time coordinate $t_0$ of
$\IC(\Cm{1},\Cm{2})$ satisfies the constraints $t'\leq t \leq t''$
and $\hat{t}\leq t \leq \check{t}$:
$$\displaylines{\psi_t\left(t_1,x_1,y_1,v_1,t_3,x_3,y_3,v_2,t',t'',\hat{t},\check{t}\right):=
\hfill{}\cr\hfill{}
t'(v_1+v_2)\leq\sqrt{(x_1-x_3)^2+(y_1-y_3)^2}+t_1v_1+t_3v_2\leq
t''(v_1+v_2) \cr\hfill{}\land \
\hat{t}(v_1+v_2)\leq\sqrt{(x_1-x_3)^2+(y_1-y_3)^2}+t_1v_1+t_3v_2\leq
\check{t}(v_1+v_2)\ .}$$

Now, $\IC(\Cm{1},\Cm{2})\subset\bB{1}\cap\bB{2}$ if and only if $\psi_{{\sf
\IC}^-}(t_1,\ab x_1,\ab y_1,\ab t_2,\ab x_2,\ab y_2,\ab v_1,\ab t_3,\ab x_3,\ab y_3,\ab t_4,\ab x_4,\ab y_4,\ab \ab
v_2)$ where $\psi_{{\sf
\IC}^-}(t_1,\ab x_1,\ab y_1,\ab t_2,\ab x_2,\ab y_2,\ab v_1,\ab t_3,\ab x_3,\ab y_3,\ab t_4,\ab x_4,\ab y_4,\ab v_2):=
\psi_t(t_1,\ab x_1,\ab y_1,\ab v_1,\ab \ab t_3,\ab x_3,\ab y_3,\ab v_2,\ab t_1,\ab t_2,\ab t_3,\ab t_4)\ \land \
 \phi_-(t_1,\ab x_1,\ab y_1,\ab v_1,\ab t_3,\ab x_3,\ab y_3,\ab v_2,\ab t_1,\ab x_1,\ab y_1,\ab t_2,\ab x_2,\ab y_2,\ab v_1)\
\land \
\phi_-(t_1,\ab x_1,\ab y_1,\ab v_1,\ab t_3,\ab x_3,\ab y_3,\ab v_2,\ab t_3,\ab x_3,\ab y_3,\ab t_4,\ab x_4,\ab y_4,\ab v_2)$
and $\IC(\Cp{1},\Cp{2})\subset\uB{1}\cap\uB{2}$ if and only if $\psi_{{\sf
\IC}^+}(t_1,\ab x_1,\ab y_1,\ab t_2,\ab x_2,\ab y_2,\ab v_1,\ab
t_3,\ab x_3,\ab \ab y_3,\ab t_4,\ab x_4,\ab y_4,\ab v_2):=
\psi_t(t_2,\ab x_2,\ab y_2,\ab v_1,\ab t_4,\ab x_4,\ab y_4,\ab
v_2,\ab t_1,\ab t_2,\ab t_3,\ab t_4)\ \land \
 \phi_-(t_2,\ab x_2,\ab y_2,\ab v_1,\ab t_4,\ab \ab x_4,\ab y_4,
 \ab v_2,\ab t_2,\ab x_2,\ab y_2,\ab \ab t_1,\ab x_1,\ab y_1,\ab v_1)\
\land \ \phi_-(t_2,\ab x_2,\ab y_2,\ab v_1,\ab t_4,\ab x_4,\ab
y_4,\ab v_2,\ab t_4,\ab x_4,\ab y_4,\ab t_3,\ab x_3,\ab \ab y_3,\ab
v_2)$.

 The
formula that expresses the criterium for Case \textbf{(III)}  then looks as
follows:
$$\displaylines{\Phi_{\textbf{III}}\left(t_1,x_1,y_1,t_2,x_2,y_2,v_1,t_3,x_3,y_3,t_4,x_4,y_4,v_2\right)
:=\left(v_1+ v_2 \neq  0\right)
\ \land\cr \lnot \
\Phi_{\textbf{I}}\left(t_1,x_1,y_1,t_2,x_2,y_2,v_1,t_3,x_3,y_3,t_4,x_4,y_4,v_2\right) \ \land \  \cr\hfill{} \left(\psi_{{\sf
\IC}^-}(t_1,x_1,y_1,t_2,x_2,y_2,v_1,t_3,x_3,y_3,t_4,x_4,y_4,v_2)\
\lor  \right.\hfill{}\cr\hfill{} \left.\psi_{{\sf
\IC}^+}(t_1,x_1,y_1,t_2,x_2,y_2,v_1,t_3,x_3,y_3,t_4,x_4,y_4,v_2)\right)\
. \hfill{}}$$

\subsection{The formula for the parametric alibi query}\label{subsec:solution}
The final formula that
decides if two beads,
$\B_1=\bead{t_1}{x_1}{y_1}{t_2}{x_2}{y_2}{v_1}$ and
$\B_2=\bead{t_3}{x_3}{y_3}{t_4}{x_4}{y_4}{v_2}$, do not intersect
looks as follows
$$\displaylines{\psi_{alibi}\left(t_1,x_1,y_1,t_2,x_2,y_2,v_1,t_3,x_3,y_3,t_4,x_4,y_4,v_2\right):=
\lnot\left((t_1<t_2\ \land \ t_3<t_4)\ \land \hfill{}\right.
\cr\quad
\left(\Phi_{\textbf{III}}
\left(t_1,x_1,y_1,t_2,x_2,y_2,v_1,t_3,x_3,y_3,t_4,x_4,y_4,v_2\right)\right.
\hfill{}\cr\hfill{} \left.\lor\ \Phi_{\textbf{II}}
\left(t_1,x_1,y_1,t_2,x_2,y_2,v_1,t_3,x_3,y_3,t_4,x_4,y_4,v_2\right)\right)\hfill{}\cr\hfill{}\left.
\lor \
\Phi_{\textbf{I}}\left(t_1,x_1,y_1,t_2,x_2,y_2,v_1,t_3,x_3,y_3,t_4,x_4,y_4,v_2\right)\right).\quad }$$

\section{Experiments}\label{sec:experiments}
In this section, we compare our solutionto the alibi query (using the formula given in Section~\ref{subsec:solution}) with the method of eliminating quantifiers of Mathematica.

In the following table it is clear that traditional quantifier elimination performs badly on the example beads. Its running times highly deviates from their average  and range in the minutes. Whereas the method described in this paper performs in running times that consistently only needs milliseconds or less. This shows our method is efficient and our claim, that it runs in milliseconds or less, holds.

For this first set of beads we chose to verify intersection of two oblique beads (1-2) and the intersection of one oblique and one straight bead (3-4). The beads that actually intersected had a remarkable low running time with the QE-method.

\medskip

\begin{tabular}{|l||c|c||c|c|}
\hline
 &\multicolumn{2}{c|}{The beads}&\multicolumn{2}{c|}{The running times}\\
\cline{2-5}
 &$\B_1$&$\B_2$&QE&Our Method\\
\hline\hline
1&$\ (0,0,0,2,0,2,1.9)\ $&$\ (0,3,0,2,3,2,2)\ $&$\ 0.656$ Seconds &$\ 0.016$ Seconds\\
2&$\ (0,0,0,2,0,2,1.9)\ $&$\ (0,4,0,2,4,2,2)\ $&$\ 324.453 $ Seconds &$\ 0.063$ Seconds\\
3&$\ (0,0,0,2,0,2,1.9)\ $&$\ (0,3,0,2,3,0,2)\ $&$\ 0.438$ Seconds &$\ 0.015$ Seconds\\
4&$\ (0,0,0,2,0,2,1.9)\ $&$\ (0,4,0,2,4,0,2)\ $&$\ 475.719$ Seconds &$\ 0.031$ Seconds\\
\hline
\end{tabular}

\medskip

The type of beads in this second set are as in the first. However, these beads all have overlapping time intervals unlike the first set, where the time intervals coincided.

\medskip

\begin{tabular}{|l||c|c||c|c|}
\hline
 &\multicolumn{2}{c|}{The beads}&\multicolumn{2}{c|}{The running times}\\
\cline{2-5}
 &$\B_1$&$\B_2$&QE&Our Method\\
\hline\hline
1&$\ (0,0,0,2,0,2,1.9)\ $&$\ (1,3,0,3,3,2,2)\ $&$\ 63.375$ Seconds &$\ 0.078$ Seconds\\
2&$\ (0,0,0,2,0,2,1.9)\ $&$\ (1,4,0,3,4,2,2)\ $&$\ 59.485$ Seconds &$\ 0.078$ Seconds\\
3&$\ (0,0,0,2,0,2,1.9)\ $&$\ (1,3,0,3,3,0,2)\ $&$\ 29.734$ Seconds &$\ 0.031$ Seconds\\
4&$\ (0,0,0,2,0,2,1.9)\ $&$\ (1,4,0,3,4,0,2)\ $&$\ 27.281$ Seconds &$\ 0.032$ Seconds\\
\hline
\end{tabular}

\medskip

The type of beads in this third set are as in the first. But this time the time intervals are completely disjoint. Note that the running times for the QE-method are more consistent in this set and the previous one.

\medskip

\begin{tabular}{|l||c|c||c|c|}
\hline
 &\multicolumn{2}{c|}{The beads}&\multicolumn{2}{c|}{The running times}\\
\cline{2-5}
 &$\B_1$&$\B_2$&QE&Our Method\\
\hline\hline
1&$\ (0,0,0,2,0,2,1.9)\ $&$\ (3,3,0,4,3,2,2)\ $&$\ 63.641$ Seconds &$\ 0.046$ Seconds\\
2&$\ (0,0,0,2,0,2,1.9)\ $&$\ (3,4,0,4,4,2,2)\ $&$\ 61.781$ Seconds &$\ 0.016$ Seconds\\
3&$\ (0,0,0,2,0,2,1.9)\ $&$\ (3,3,0,4,3,0,2)\ $&$\ 52.735$ Seconds &$\ 0.031$ Seconds\\
4&$\ (0,0,0,2,0,2,1.9)\ $&$\ (3,4,0,4,4,0,2)\ $&$\ 56.875$ Seconds &$\ 0.046$ Seconds\\
\hline
\end{tabular}

\section{The alibi query at a fixed moment in time}\label{sec:4circles}
\subsection{Introduction}
In this section, we present another example where common sense prevails over the general quantifier-elimination methods. The problem is the following. As in the previous setting, we have lists of time stamped-locations of two moving objects and upper bounds on the object's speed between time stamps. We wish to know if two objects could have met at a given moment in time.

For the remainder of this section, we reuse the assumptions from the previous section. We wish to verify if the
beads $\B_1=\bead{t_1}{x_1}{y_1}{t_2}{x_2}{y_2}{v_1}$ and
$\B_2=\bead{t_3}{x_3}{y_3}{t_4}{x_4}{y_4}{v_2}$ intersect at a moment in time $t_0$. Moreover,
we assume the beads are non-empty, i.e., $(x_2-x_1)^2+(y_2-y_1)^2
\leq (t_2-t_1)^2v_1^2$ and $(x_4-x_3)^2+(y_4-y_3)^2 \leq
(t_4-t_3)^2v_2^2$ and that $t_1\leq t_0 \leq t_2$ and $t_3\leq t_0 \leq t_4$ are satisfied.
This means we need to eliminate the quantifiers in
$$\displaylines{\exists x\exists y \left( (x-x_1)^2+(y-y_1)^2\leq v_1^2(t_0-t_1)^2\ \land \ (x-x_2)^2+(y-y_2)^2\leq v_1^2(t_0-t_2)^2\right.
\cr\hfill{} \left.
\land \ (x-x_3)^2+(y-y_3)^2\leq v_2^2(t_0-t_3)^2\ \land \ (x-x_4)^2+(y-y_4)^2\leq v_2^2(t_0-t_4)^2\right) .}
$$

Eliminating quantifiers gives us a formula that decide whether or not four discs have a non-empty intersection. For ease of notation we will use the following abbreviations: $(x,y)\in D_i$ if and only if $(x-x_i)^2+(y-y_i)\leq r_i^2$ and $(x,y)\in C_i$ if and only if $(x-x_i)^2+(y-y_i)= r_i^2$.

\subsection{Main theorem}

Using Helly's theorem we can simplify the problem even more. Helly's theorem states that if you have a set $S$ of $m$ convex sets in $n$ dimensional space and if any subset of $S$ of $n+1$ convex sets has a non-empty intersection, then all $m$ convex sets have a non-empty intersection. For the plane, this means we only need to find a quantifier free-formula that decides if three discs have a non-empty intersection. For the remainder of this section assume that we want to verify whether $D_1\cap D_2\cap D_3$ is non-empty.

\begin{theorem}\label{littleconjecture}
Three discs, $D_1$, $D_2$ and $D_3$, have a non-empty intersection if and only if one of the following cases occur:
\begin{enumerate}
\item there is  a disc whose center is in the other two discs; or
\item the previous case does not occur and there exists a pair of discs for which one of both intersection points of their bordering circles lies in the remaining disc.
\end{enumerate}
\end{theorem}

\begin{proof}
The \emph{if}-direction is trivial. The \emph{only if}-direction is less trivial.
We will use the following abbreviations, $D=D_1\cap D_2\cap D_3$ and $C=\partial D$.

Assume $D$ is non-empty and that neither (1) nor (2) holds. The intersection $D$ is convex as it is the intersection of convex sets. We distinguish between the case where $D$ is  a point or and the case where $D$ is not a point.
\begin{itemize}
\item Suppose $D$ is a single point $p$. This point $p$ can not lie in the interior of the three discs, because $D$ would not be a point then.

    Nor can $p$ lie in the interior of two discs. If that would be the case then there exists a neighborhood of $p$ that is part of the intersection of those two discs, say $D_1$ and $D_2$. Moreover $p$ would be part of $C_3$ and this neighborhood would intersect the interior of $D_3$. This means $D$ is not a point.

    So $p$ must lie on the border of two discs, say $D_1$ and $D_2$, and $p$ must also be part of $D_3$ because $D=\{p\}$. This contradicts our assumption that (2) does not hold.
\item Assume $D$ is not a point. All points on $C$ belong to at least one $C_i$. If there is a point that does not belong to any $C_i$, then it is in the interior of all $D_i$ and there exists a neighborhood of that point that is in the interior of all $D_i$ and hence in $D$. That contradicts to the fact that this point is in $C$.

    Furthermore, not all points of $C$ belong to a single $C_i$. If that was the case then $D_i$ would be part of (and equal to) $D$ and its center would be inside the other two discs which contradicts the assumption that (1) does not hold.

    So, $C$ is made up of parts of the $C_i$, of which some may coincide but not all of them. When traveling along $C$ you will encounter a point $p$ that connects part of a $C_i$ and part of a $C_j$, where $i\neq j$, that do not coincide, otherwise (1) must occur again which is a contradiction. However, this $p$ also yields to a contradiction since it belongs to two different $C_i$, say $C_1$ and $C_2$, and is part of $C$ hence $D$ and $D_3$. This contradicts the assumption that (2) does not occur.\qed
\end{itemize}
\end{proof}

\subsection{Translating the theorem in a formula}\label{4circlestoformula}

We can simplify the equations even further using coordinate transformations. By applying a translation, rotation and scaling we may assume that $(x_1,y_1)=(0,0)$, $x_2\geq 0$, $r_1=1$ and $y_2=0$.
Using these simplifications and translating Theorem~\ref{littleconjecture}, we get the following formula.
$$\displaylines{\Psi_1(x_2,r_2,x_3,y_3,r_3):=
\left((-x_2)^2\leq r_2^2\ \land \ (-x_3)^2+(-y_3)^2\leq r_3^2\right)\ \lor \hfill{}\cr\hfill{}  \exists x\exists y\left( x^2+y^2=1\ \land \ (x-x_2)^2+y^2=r_2^2\ \land \ (x-x_3)^2+(y-y_3)^2\leq r_3^2\right)\ .}
$$
This is a formula that decides if either the center of the first disc is part of the two other discs, see the first line, or if either there exists a point in the intersection of the first two circles that is part of the third disc. All that remains now is making the expression $$\exists x\exists y \left( x^2+y^2=1\ \land \ (x-x_2)^2+y^2=r_2^2\ \land \ (x-x_3)^2+(y-y_3)^2\leq r_3^2\right)$$ quantifier free.

To do this we assume that $C_1$ and $C_2$ do not coincide but have a non-empty intersection. This is equivalent to $x_2\neq 0 \ \land \ x_2\leq r_2+1$. Next, we need to compute the point(s) where $C_1$ and $C_2$ intersect.

$$\left\{
\begin{array}{l}
x^2+y^2 = 1\\
(x-x_2)^2+ y^2 = r_2^2
\end{array}
\right.
\mbox{ or }\left\{
\begin{array}{l}
x=\frac{x_2^2+1-r_2^2}{2x_2}\\
y=\pm\sqrt{1-\left(\frac{x_2^2+1-r_2^2}{2x_2}\right)^2}
\end{array}
\right.
$$
$$
\mbox{ or }\left\{
\begin{array}{l}
x=\frac{x_2^2+1-r_2^2}{2x_2}\\
y=\pm\sqrt{\left(1-\frac{x_2^2+1-r_2^2}{2x_2}\right)\left(1+\frac{x_2^2+1-r_2^2}{2x_2}\right)}
\end{array}
\right.
$$
$$\mbox{ or }\left\{
\begin{array}{l}
x=\frac{x_2^2+1-r_2^2}{2x_2}\\
y=\pm\frac{1}{2x_2}\sqrt{\left(r_2^2-\left(x_2-1\right)^2\right)\left(\left(1+x_2\right)^2-r_2^2\right).}
\end{array}
\right.
$$
If $y=0$, then verifying if that single point of intersection is part of $D_3$ is easy, one only needs to verify if
$$ \left(\frac{x_2^2+1-r_2^2}{2x_2} - x_2\right)^2+y_3^2\leq r_3^2$$
$$\mbox{ or equivalently }\left(1-r_2^2 - x_2^2\right)^2+4x_2^2y_3^2\leq 4x_2^2r_3^2\ .$$
If $y\neq0$, then verifying if one both points of intersection is part of $D_3$ is less trivial, since this involves square roots
$$ \left(\frac{x_2^2+1-r_2^2}{2x_2} - x_3\right)^2+\left(\pm\frac{1}{2x_2}\sqrt{\left(r_2^2-\left(x_2-1\right)^2\right)
\left(\left(1+x_2\right)^2-r_2^2\right)}-y_3\right)^2\leq r_3^2$$
$$\left(x_2^2+1-r_2^2-2x_2x_3\right)^2+\left(\pm\sqrt{\left(r_2^2-\left(x_2-1\right)^2\right)
\left(\left(1+x_2\right)^2-r_2^2\right)}-2x_2y_3\right)^2\leq 4x_2^2r_3^2$$
$$\displaylines{\mbox{ or  }\left(x_2^2+1-r_2^2-2x_2x_3\right)^2+\left(r_2^2-\left(x_2-1\right)^2\right)\left(\left(1+x_2\right)^2-r_2^2\right)
\cr+(2x_2y_3)^2\pm4x_2y_3\sqrt{\left(r_2^2-\left(x_2-1\right)^2\right)\left(\left(1+x_2\right)^2-r_2^2\right)}\leq 4x_2^2r_3^2}$$
$$\displaylines{\mbox{ or }\left(x_2^2+1-r_2^2-2x_2x_3\right)^2+\left(r_2^2-\left(x_2-1\right)^2\right)\left(\left(1+x_2\right)^2-r_2^2\right)
\cr+(2x_2y_3)^2-4x_2^2r_3^2\leq \pm4x_2y_3\sqrt{\left(r_2^2-\left(x_2-1\right)^2\right)\left(\left(1+x_2\right)^2-r_2^2\right)}.}$$

This is almost a \fop-formula except for the square root. However, the square root can be eliminated as we will show next.
The previous expression is of the form $L\leq \pm a\sqrt{W}$. The presence of the $\pm$ simplifies this a lot, this means either sign of the square root will do, and also that we may assume the right hand-side is positive. Of course the square root must exist as well, this means $W\geq 0$.

This expression can then be simplified to
$$W\geq 0 \ \land \ \left(L\leq 0 \ \lor \ L^2\leq a^2W\right)$$
and gives us the expression

$$\displaylines{\Phi_2(x_2,r_2,x_3,y_3,r_3):=\left(r_2^2-\left(x_2-1\right)^2\right)\left(\left(1+x_2\right)^2-r_2^2\right)\geq 0 \hfill{}\cr\hfill{}\land \
\left(\left(x_2^2+1-r_2^2-2x_2x_3\right)^2+\left(r_2^2-\left(x_2-1\right)^2\right)\left(\left(1+x_2\right)^2-r_2^2\right)
+(2x_2y_3)^2\right.
\cr\hfill{}-4x_2^2r_3^2\leq 0 \ \lor\ \left(
\left(x_2^2+1-r_2^2-2x_2x_3\right)^2+\left(r_2^2-\left(x_2-1\right)^2\right)\left(\left(1+x_2\right)^2-r_2^2\right)
\right.\cr\hfill{}\left.\left.+(2x_2y_3)^2-4x_2^2r_3^2\right)^2\leq \left(4x_2y_3\right)^2\left(r_2^2-\left(x_2-1\right)^2\right)\left(\left(1+x_2\right)^2-r_2^2\right)\right).}$$

\subsection{The safety formula}

Now, all that remains  to be constructed is a formula that returns the convenient coordinates and a formula that guarantees that $C_1$ and $C_2$ actually intersect for safety, i.e., to exclude the case of empty intersection. The latter is constructed as follows. The formula $\varphi(x_1,y_1,r_1,x_2,y_2,r_2)$ returns {{\sc True}} if and only if the two circles, with centers $(x_1,y_1)$ and $(x_2,y_2)$ and radii $r_1$ and $r_2$ respectively, have a distance between their centers that is not larger that the sum of their radii and not equal to zero to ensure they do not coincide. We have
$$\varphi(x_1,y_1,r_1,x_2,y_2,r_2):=0<(x_2-x_1)^2+(y_2-y_1)^2\leq (r_1+r_2)^2\ .$$
The formula $\phi(x_1,y_1,r_1,x_2,y_2,r_2)$ returns {{\sc True}} if and only if the second circle is not fully enclosed by the first, i.e., the sum of the distance between the centers plus the second radius is bigger than the first radius and vice versa.
We can write
$$\phi(x_1,y_1,r_1,x_2,y_2,r_2):=(x_2-x_1)^2+(y_2-y_1)^2\geq (r_1-r_2)^2\ .$$
These two safety conditions give us our safety formula
$$\Phi_{\mbox{safety}}(x_1,y_1,r_1,x_2,y_2,r_2):=\varphi(x_1,y_1,r_1,x_2,y_2,r_2)\ \land \ \phi(x_1,y_1,r_1,x_2,y_2,r_2)\ .$$

\subsection{The change of coordinates}
The transformation consists of a translation, rotation and scaling. The translation to move the first circle's center to the origin. The rotation to align the second center with the $x$-axis. Finally the scaling to ensure that the first circle's radius is equal to one.
First, the translation $T(x,y):=(x-x_1,y-y_1).$
The rotation
 is $$ R(x,y):=\frac{1}{\sqrt{(x_2-x_1)^2+(y_2-y_1)^2}}
\left(
\begin{array}{cc}
x_2-x_1&y_2-y_1\\
y_1-y_2&x_2-x_1
\end{array}
\right)
\left(
\begin{array}{c}
x\\
y
\end{array}
\right)\ $$
and finally, the scaling is
$S(x,y):=\frac{1}{r_1}(x,y).$
The transformation is then a composition of those three transformations $A(x,y)=(S\circ R \circ T)(x,y):=$
$$\left(\frac{(x_2-x_1)(x-x_1)+(y_2-y_1)(y-y_1)}{r_1\sqrt{(x_2-x_1)^2+(y_2-y_1)^2}},
\frac{(y_1-y_2)(x-x_1)+(x_2-x_1)(y-y_1)}{r_1\sqrt{(x_2-x_1)^2+(y_2-y_1)^2}}\right)\ .$$
The following formula takes three circles with centers $(x_i,y_i)$ and radii $r_i$ respectively, and transforms them in three new circles where the first circle has center $(0,0)$ and radius 1, the second circle has center $(x'_2,0)$ and radius $r'_2$ and the third circle has center $(x'_3,y'_3)$ and radius $r'_3$.
$$\displaylines{\Phi_{\mbox{transformation}}(x_1,y_1,r_1,x_2,y_2,r_2,x_3,y_3,r_3,x'_2,r'_2,x'_3,y'_3,r'_3):=
\hfill{}\cr
\hfill{}x'_2=\frac{\sqrt{(x_2-x_1)^2+(y_2-y_1)^2}}{r_1}\ \land \ r'_2=\frac{r_2}{r_1}\ \land \ r'_3=\frac{r_3}{r_1}\ \land\cr
\hfill{}x'_3=\frac{(x_2-x_1)(x_3-x_1)+(y_2-y_1)(y_3-y_1)}{r_1\sqrt{(x_2-x_1)^2+(y_2-y_1)^2}}\ \land \cr
\hfill{}y'_3=\frac{(y_1-y_2)(x_3-x_1)+(x_2-x_1)(y_3-y_1)}{r_1\sqrt{(x_2-x_1)^2+(y_2-y_1)^2}}\ \ .}$$
Note that this is not a \fop-formula anymore due to the square roots and fractions. This "formula" is meant to act like a function, which substitutes coordinates. The substituted coordinates have fractions and square roots but these can easily be disposed of when having the entire inequality on a common denominator, isolating the square root and squaring the inequality, as we showed in Section~\ref{4circlestoformula}.

\subsection{The  formula for the alibi query at a fixed moment in time}
First, we construct a formula that checks for any of two circles out of three if any of the conditions in Theorem~\ref{littleconjecture} are satisfied.
$$\displaylines{\Psi_{\mbox{2/3}}(x_1,y_1,r_1,x_2,y_2,r_2,x_3,y_3,r_3)\ :=\ \exists x'_2\exists r'_2\exists x'_3\exists y'_3\exists r'_3
\left(
\right.
\hfill{}\cr
\hfill{}\Phi_{\mbox{transformation}}(x_1,y_1,r_1,x_2,y_2,r_2,x_3,y_3,r_3,x'_2,r'_2,x'_3,y'_3,r'_3)\ \land \cr\hfill{}
\left(\Psi_1(x'_2,r'_2,x'_3,y'_3,r'_3)\ \lor \ \Phi_{\mbox{safety}}(x_1,y_1,r_1,x_2,y_2,r_2)\ \land\ \Phi_2(x'_2,r'_2,x'_3,y'_3,r'_3)\right)
\cr\lor \hfill{}
\Phi_{\mbox{transformation}}(x_1,y_1,r_1,x_3,y_3,r_3,x_2,y_2,r_2,x'_2,r'_2,x'_3,y'_3,r'_3)\ \land \cr\hfill{}
\left(\Psi_1(x'_2,r'_2,x'_3,y'_3,r'_3)\ \lor \ \Phi_{\mbox{safety}}(x_1,y_1,r_1,x_3,y_3,r_3)\ \land\ \Phi_2(x'_2,r'_2,x'_3,y'_3,r'_3)\right)
\cr\lor \hfill{}
\Phi_{\mbox{transformation}}(x_2,y_2,r_2,x_3,y_3,r_3,x_1,y_1,r_1,x'_2,r'_2,x'_3,y'_3,r'_3)\ \land \cr\hfill{}
\left.
\left(\Psi_1(x'_2,r'_2,x'_3,y'_3,r'_3)\ \lor \ \Phi_{\mbox{safety}}(x_2,y_2,r_2,x_3,y_3,r_3)\ \land\ \Phi_2(x'_2,r'_2,x'_3,y'_3,r'_3)\right)\right)}$$

This formula is all we need to incorporate Helly's theorem in our final formula. Four discs have a non-empty intersection if and only if the following formula is satisfied
$$\displaylines{\Psi(x_1,y_1,r_1,x_2,y_2,r_2,x_3,y_3,r_3,x_4,y_4,r_4):=\hfill{}
\cr\hfill{}\Psi_{\mbox{2/3}}(x_1,y_1,r_1,x_2,y_2,r_2,x_3,y_3,r_3)\ \land\ \Psi_{\mbox{2/3}}(x_1,y_1,r_1,x_2,y_2,r_2,x_4,y_4,r_4)\  \hfill{}
\cr\hfill{}\land\ \Psi_{\mbox{2/3}}(x_1,y_1,r_1,x_3,y_3,r_3,x_4,y_4,r_4)\  \land\ \Psi_{\mbox{2/3}}(x_2,y_2,r_2,x_3,y_3,r_3,x_4,y_4,r_4).\quad}$$

This is almost a quantifier free-formula except for the fractions and square roots. However, as we showed before these can easily be disposed of. We omitted these tedious conversions for the sake of clarity.

\section{Conclusion}
In this paper, we proposed a method that decides if two beads have a
non-empty intersection or not. Existing quantifier-elimination methods could achieve this
already through means of quantifier elimination though not in a
reasonable amount of time. Deciding intersection of concrete beads
took of the order of minutes, while the parametric case could be
measured at least in days if a solution would ever be obtained. The parametric solution we laid out in
this paper only takes a few milliseconds or less.

The solution we present is a first-order formula containing square root-expressions. These can easily be disposed of using repeated squarings and adding extra conditions, thus obtaining a true quantifier free-expression for the alibi query.

We also give a  solution to the alibi query at a fixed moment in time.

The solutions we propose are based on geometric argumentation and they
illustrate the fact that some practical problems require  creative solutions, where at least in theory, existing systems could provide a solution.

\section*{Acknowledgements}
  This research has been partially funded by the European Union under the FP6-IST-FET programme,
  Project n. FP6-14915, GeoPKDD: Geographic Privacy-Aware Knowledge Discovery and Delivery, and  by the Research Foundation Flanders (FWO-Vlaanderen), Research Project
G.0344.05.

\bibliographystyle{acmtrans}\bibliographystyle{plain}

\section*{Appendix: The {\sc Mathematica} implementation}

The function \textbf{alibi}, using the method described in this paper, returns {{\sc True}} if the bead with apexes $(t1,x1,y1)$ and $(t2,x2,y2)$ and speed $v1$ intersects the bead with apexes $(t3,x3,y3)$ and $(t4,x4,y4)$ and speed $v2$ and {\sf False} otherwise.
The function \textbf{alibiQE} does the same except it uses the built-in quantifier elimination method of {\sc Mathematica}.

Note that Case \textbf{(II)} in our implementation corresponds to Case \textbf{(III)} in the description and vice versa. The reason for doing so is that Case \textbf{(III)}, in the description, is computationally a lot easier than Case \textbf{(II)}. Moreover, once any of the three cases returns {{\sc True}}, our implementation exits and returns that result and thus omitting further computation.

\begin{figure}[h]
\centering
\epsfig{file=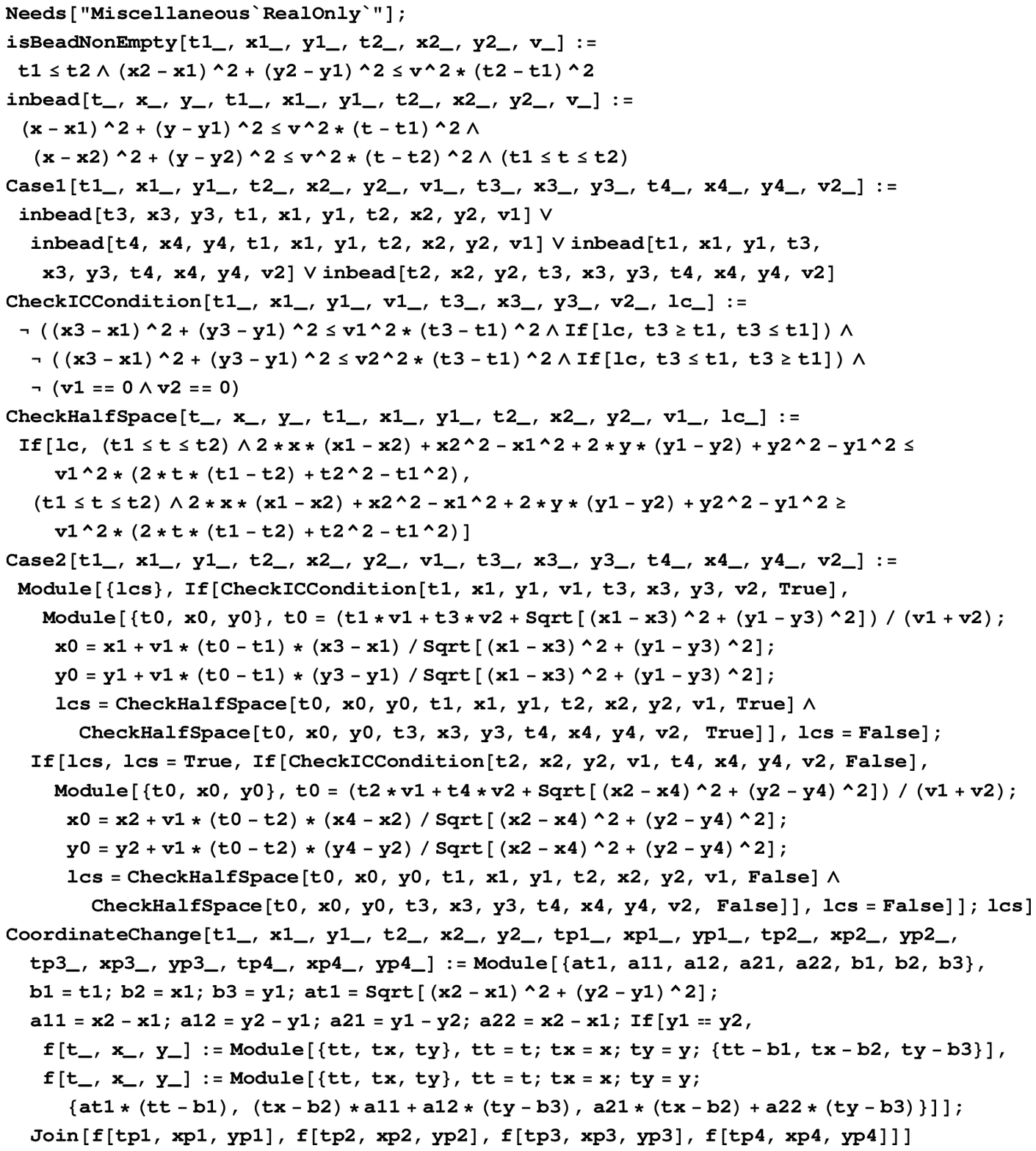, width=5in, height=7in}
\end{figure}
\clearpage
\begin{figure}
\centering
\epsfig{file=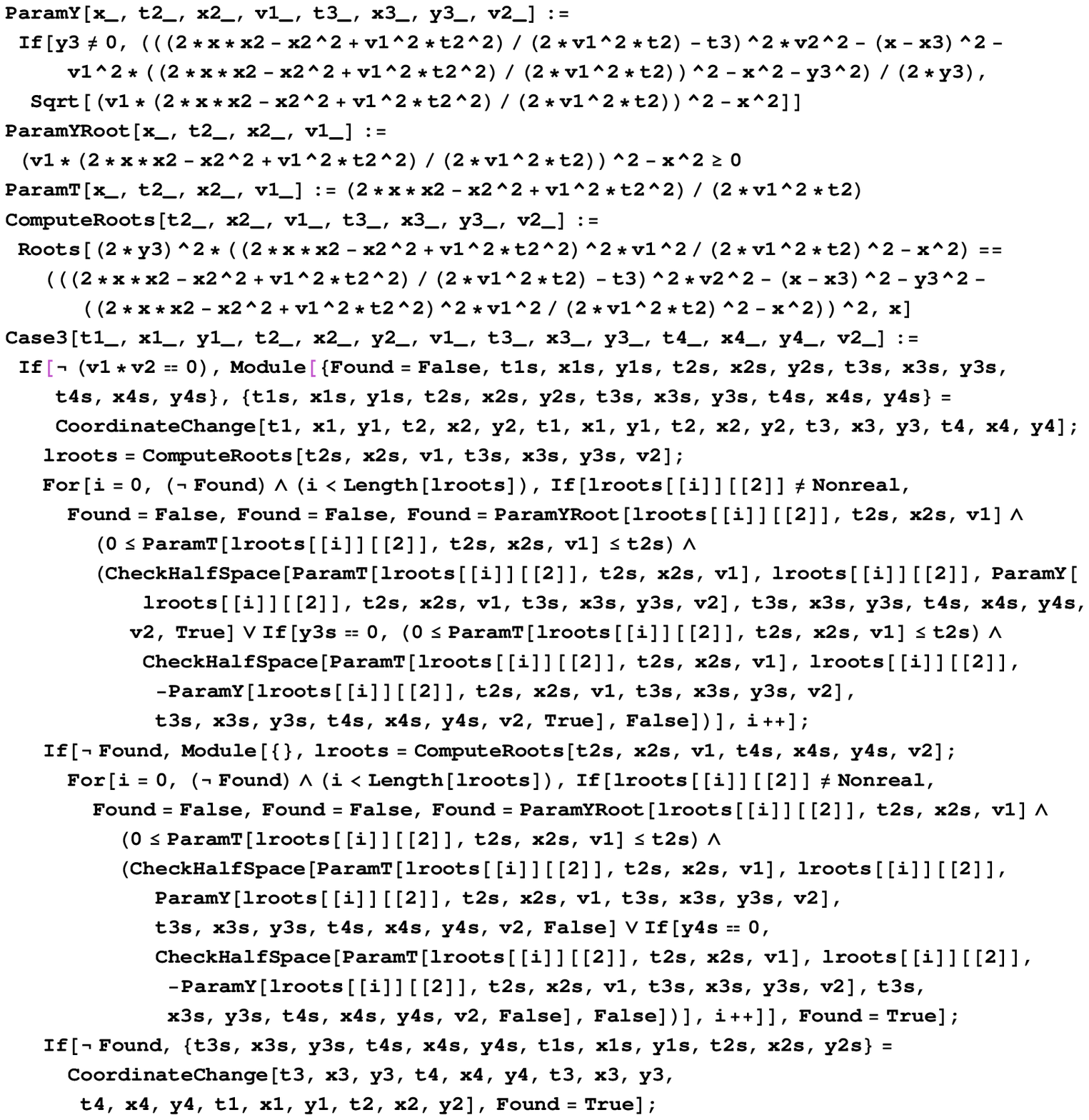, width=5in, height=7.5in}
\end{figure}
\clearpage
\begin{figure}
\centering
\epsfig{file=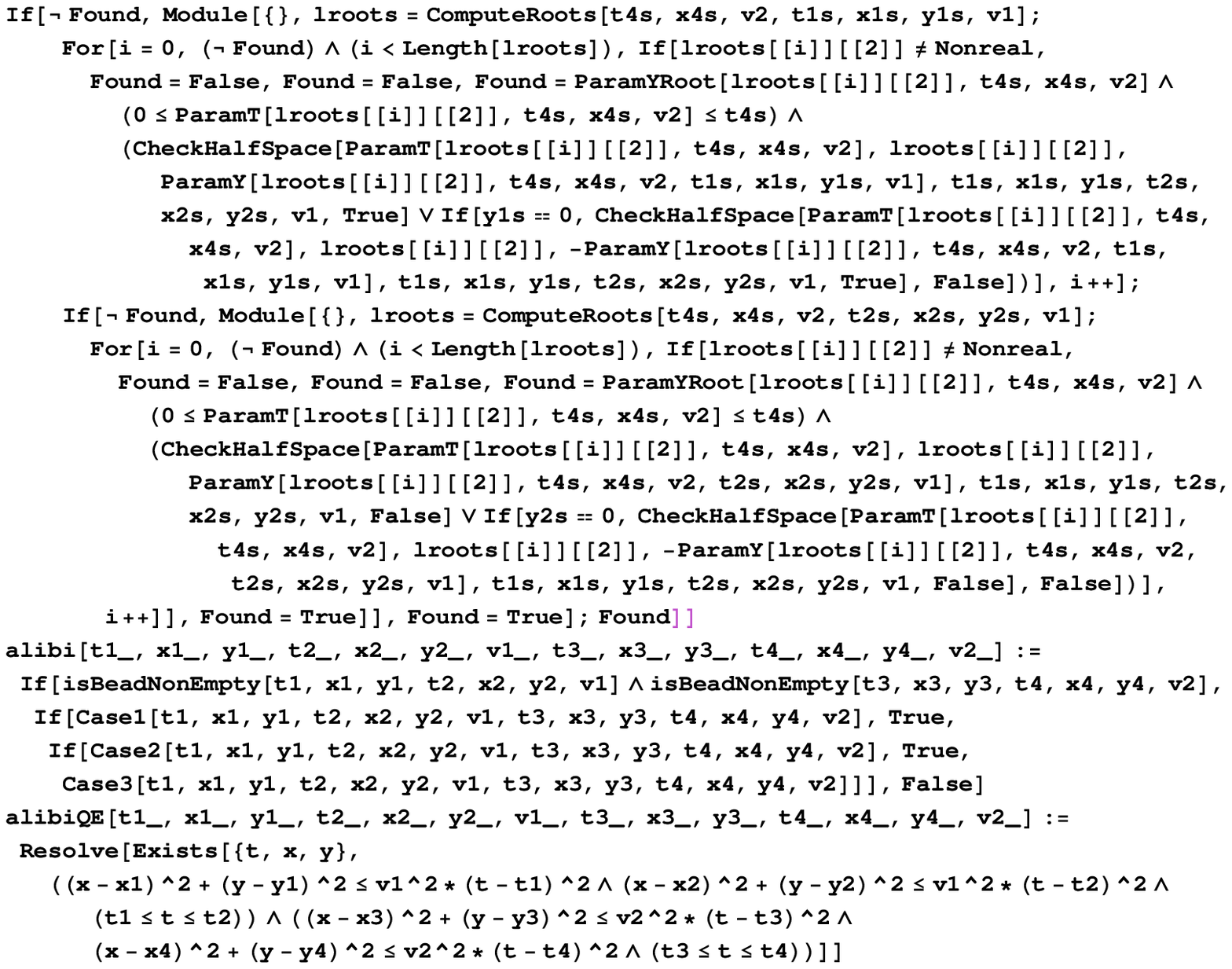, width=5in, height=5.5in}
\end{figure}

\end{document}